\documentclass[10pt]{article}

\usepackage{booktabs}
\usepackage{bm}
\usepackage{fullpage}
\usepackage{lipsum}
\usepackage{amsfonts, mathrsfs}
\usepackage{graphicx}
\usepackage{epstopdf}
\usepackage{multirow}
\usepackage{algorithm}
\usepackage{algpseudocode}
\usepackage{color}
\usepackage{amsthm}
\usepackage{amsmath}

\usepackage{amsopn}

\AtBeginDocument{%
  \providecommand\BibTeX{{%
    \normalfont B\kern-0.5em{\scshape i\kern-0.25em b}\kern-0.8em\TeX}}}

\newcommand{\NSU}{{\tt NonAdaptiveWithK($k,c$)}}
\newcommand{\NAdaptiveU}{{\tt SublinearDecrease}$(b)$}

\newcommand{\tbc}[1]{{\color{black}{#1}}}
\newcommand{\gdm}[1]{{\color{black}{#1}}}
\newcommand{\gia}[1]{{\color{black}{#1}}}

\newcommand{\darek}[1]{{\color{black}{#1}}}
\newcommand{\dk}[1]{{\color{black}{#1}}}
\newcommand{\gs}[1]{{\color{black}{#1}}}

\newtheorem{fact}{Fact}[section]
\newcommand{\cE}{{\mathcal E}}
\newcommand{\E}{\mathbf{E}}
\newcommand{\cA}{{\mathcal A}}
\newcommand\given[1][]{\:#1\vert\:}
\newcommand{\hsigma}{{\hat{\sigma}}}
\newcommand{\hA}{{\hat{A}}}
\newcommand{\tc}{{\tt time\_counter}} 

\newcommand{\MLineComment}[1]
{
	/*
	\begin{minipage}[t]{.60\textwidth}
		{\it {#1}} */
	\end{minipage}
	
}

\begin{document}

\title{Time and Energy Efficient Contention Resolution in Asynchronous Shared Channels
}


\author{%
	Gianluca De Marco\footnotemark[4]
	\and
	Dariusz R.~Kowalski\footnotemark[3]
	\and
	Grzegorz Stachowiak\footnotemark[2]
}

\date{}

\maketitle

\footnotetext[4]{
Dipartimento di Informatica,
University of Salerno,
Italy.
Email: \texttt{gidemarco@unisa.it}.
}

\footnotetext[2]{
Institute of Computer Science, University of Wroc{\l}aw, Poland.
Email: \texttt{gst@cs.uni.wroc.pl}.
}

\footnotetext[3]{School of Computer and Cyber Sciences, Augusta University, GA, USA.
Email: \texttt{dkowalski@augusta.edu}}

\maketitle

\newtheorem{theorem}{Theorem}[section]
\newtheorem{lemma}{Lemma}[section]
\newtheorem{corollary}{Corollary}[section]
\newtheorem{claim}{Claim}[section]
\newtheorem{proposition}{Proposition}[section]
\newtheorem{definition}{Definition}[section]
\newtheorem{example}{Example}[section]

\begin{abstract}
A shared channel (also called a multiple access channel), introduced nearly 50 years ago, is among the most popular and widely studied models of communication and distributed computing.
In a nutshell, a number of stations, \gia{independently activated over time}, is able to communicate by transmitting and listening to a shared channel \gia{in discrete time slots},
and a message is successfully delivered to all stations if and only if its source station is the only transmitter at a time. 
Despite a vast amount of work in the last decades, 
many fundamental questions remain open 
\gia{in the realistic situation
where stations do not start synchronously but are 
awaken in arbitrary times (called dynamic or \textit{asynchronous} scenario).}
What is the impact of an \gia{asynchronous start} on channel utilization?
How important is the knowledge/estimate of the number of contenders?
Could non-adaptive protocols be asymptotically as efficient as adaptive ones? 
In this work we present a broad picture of 
results answering the abovementioned questions
for the fundamental problem of Contention resolution, in which each of the contending stations needs to broadcast successfully its message.

We show that adaptive algorithms or algorithms with the knowledge of the
contention size $k$ achieve a linear $O(k)$ message latency 
even if the channel feedback is restricted to simple acknowledgements in case of successful transmissions 
and in the absence of synchronization.
This asymptotically optimal performance cannot be extended to other settings:
we prove that there is no non-adaptive algorithm without the knowledge of contention size $k$ admitting latency $o(k\log k/(\log\log k)^2)$.
This means, in particular, that coding (even random) with acknowledgements 
is not very efficient
on a shared channel without synchronization or an estimate of the contention size.
\gia{We also present a non-adaptive algorithm with no knowledge of contention size that almost matches the lower bound on latency.}

Finally, despite the absence of a collision detection mechanism, we show that our algorithms
are also efficient in terms of energy, understood as the total number of transmissions performed 
by the stations during the execution.
\end{abstract}

\providecommand{\keywords}[1]{\textbf{\textit{Key words---}} #1}

\section{Introduction}

A shared channel, also called a multiple access channel, is one of the fundamental 
communication models considered in the literature. 
It allows many autonomous computing entities to communicate over a shared medium,
and the main challenge is how to efficiently resolve collisions occurring 
when more than one entity attempts to access the channel at the same time.
Despite a vast amount of work, some fundamental questions 
about channel utilization have remained open for decades. 
What is the impact of \gia{asynchrony among the starting times
of the stations}?
How important is the knowledge/estimate of the number of contenders?
Could non-adaptive protocols or random codes be asymptotically as
efficient as adaptive protocols? This paper attempts to answer these questions
by separating the impact of the abovementioned characteristics into two classes:
those allowing asymptotically optimal channel utilization and those
which incur a substantial overhead.

\medskip

The formal model considered in this paper is the one commonly taken
as the basis for theoretical studies on multiple access channels
(\textit{cf.} the surveys by Gallager~\cite{Gal} and Chlebus~\cite{Chl}).
In what follows we overview it, paying special attention to the
particular settings of this paper.

\paragraph{Stations.} 
A set of $k$ stations are connected to the same multi-point transmission medium. The stations are anonymous, that is, they have no identification label (ID) to uniquely distinguish them.
There is no central control: every station acts autonomously by means of a distributed algorithm.
\gia{At the beginning, all the stations are sleeping. 
\dk{A sleeping station does not send nor receive any message on the channel.}
Each station can be activated \dk{(and thus become active)} at any time with the task of sending a data packet.
The activation times do not depend on algorithm's execution \gs{and} are
totally determined by an adversary that will be defined next.}
\gia{A station is active from the activation time until its termination,
which means going back into a permanent sleeping mode.}
\dk{Unlike the activation time, termination is decided by the algorithm.}

\paragraph{Communication.}
\gia{Time is \gs{divided into} discrete synchronous {\em rounds} 
(also called {\em time steps} or {\em time slots}).}
\gia{In each round, an active station can either send a message
or listen to the channel. 
A message is just the packet itself, although in our adaptive algorithm 
we allow also some stations to send  a one-bit message \gs{as coordinating 
information} \dk{(also called {\em control message} or {\em control bit}, both in theoretical models and technological applications).} 

An active station that is not transmitting in a round,
is implicitly assumed to be listening in that round.
A station transmits successfully its message at a given round if and only
if it is the only transmitter \gs{in} that round: all the other 
active stations are listening (and therefore receive the message). Namely,}
if $m\leq k$ stations transmit at the same round, then the result of the transmission \gia{in this round} depends on the parameter $m$ as follows: 
\begin{itemize}
  \item If $m = 0$, the channel is {\em silent} and 
  no packet is successfully transmitted;
  \item If $m = 1$, the packet owned by the singly transmitting station is {\em successfully transmitted} on the channel and
  therefore \gia{received} by all the other \textit{active} stations;
  \item If $m > 1$, simultaneous transmissions interfere with one another (we say that a \textit{collision} occurs) and
  as a result no \gia{message} is received by the other stations.
\end{itemize}

\paragraph{Feedback.}
In the setting \textit{without collision detection} adopted in this paper, no special signal is perceived in the case of collision, making therefore impossible, \gia{for a station listening to the channel}, 
to distinguish between an occurred 
collision and the case where no station transmits.
(By contrast, in a \textit{collision detection} setting, 
not considered in this paper, the channel elicits a \textit{feedback} 
in case of collision, 
allowing to deduce that two or more stations tried to transmit at the same time.)
The only feedback a station can sense is when it actually transmits successfully, in which case it gets an \textit{acknowledgement}.

\paragraph{Contention resolution problem.}
Each \gia{active} station possesses a packet, which can be transmitted in a single time slot. \gia{Packets are not \dk{assumed to be} distinguishable, so they cannot be used to identify the stations.}
The aim is to let each of the $k$ stations to transmit successfully its packet. 
\gia{The task is considered to be accomplished
\gs{when all packets are successfully transmitted and all stations are switched off with all their functions permanently disabled.}
}

\paragraph{Algorithms.}
We seek \textit{randomized} algorithms that allow every station to transmit successfully its packet,
regardless of the activation times. Note that the choice of randomized solutions is forced 
since in the absence of unique ID's there is no deterministic way of breaking the symmetry among the \textit{identical} stations on the shared channel.

\gia{
For each station, a randomized algorithm specifies two things:
($a$) the probability of transmission for each round of its local clock 
and ($b$) the message to be transmitted \gs{in} that round. 
In \textit{non-adaptive} algorithms the probability of transmission
depends only on the local round number \dk{(we do not assume independence of these probabilities across rounds),} and 
the message is simply the data \dk{packet} 
initially assigned to the station. 
In \textit{adaptive} algorithms, both the probability of transmission and
the message to be sent \gs{in} round $i$ depend on the history of successful transmissions received until round $i$ (messages collected so far as well
as the rounds at which they were received). 

\gia{In the non-adaptive setting, a} station 
\gia{automatically} leaves the system (switches off) once it gets
the acknowledgement that its packet has been successfully transmitted. 
Notice that this is a straight consequence of the definition above. 
Indeed, in a non-adaptive algorithm there is no point in keeping alive
a station after its successful transmission, as it cannot influence 
in any way the behaviour of the other stations.
\gia{On the contrary, in} our adaptive algorithm some stations, in order to coordinate the protocol, reserve the right to
remain active \dk{some time} after \dk{their} successful transmission.

In particular, the adaptive algorithm considered in this paper 
allows stations to send either the packet itself, 
or (alternatively) a one-bit 
message as a \dk{coordination/control} information. 
Similar adaptive settings allow to append a constant number of bits to 
each packet (\textit{e.g.} \cite{ICPADS20}) or to send any packet-sized
message, besides the original packet itself 
(\textit{e.g.} \cite{AMM13,Bend-20}).
}

A \textit{contention resolution} 
algorithm is a distributed algorithm that schedules the transmissions 
\gs{in station's local time for}
each of the $k$ participating stations
guaranteeing that every station eventually transmits 
\dk{successfully}
(\textit{i.e.}, without interfering with other stations) on the channel, \gia{and switches off}.

\paragraph{Static vs dynamic scenarios.}
\gia{Most of} the literature on the contention resolution problem
produced so far either assumed the (simplified) \textit{static} situation in which the $k$ stations are all activated at the very beginning 
(and therefore start simultaneously their transmitting schedules)~\cite{AMM13,Cap,CGR,GW,GFL,GL,KG}
or that the activation times are restricted to statistical 
(\textit{e.g.}, when packet arrivals are determined by a Poisson distribution) 
or adversarial-queueing models 
\cite{Bend-05,CKR-TALG-12,Goldberg, Kumar, MB,RagUp}. 

Inspired by the inherently decentralized nature of the multiple access model, in this paper we focus on \gia{the}
more general and realistic \textit{ dynamic} scenario, in 
\gia{which stations awaken} (\textit{i.e.}, start their local executions
of a distributed algorithm) in \textit{arbitrary times}, \textit{i.e.}, 
the sequence of activation times, 
also called a \textit{wake-up schedule}, is totally determined by an adversary \darek{(see the next paragraph for a formal definition)}.
 This
scenario, sometimes also called asynchronous, reflects the more realistic situation in which 
the stations are geographically far apart or totally 
independent, and consequently each activation time is locally determined and cannot be known 
or even approximately predicted by other stations. Throughout the paper, ``switched on'', ``activated'',
``woken up'' and their derived terms are used interchangeably 
to mean the action, controlled by an adversary, by which a 
station wakes up and starts executing the algorithm.

\paragraph{Adversaries.}

\darek{A dynamic scenario could be caused by an adaptive or an oblivious adversary. The former can decide, online during the execution, what station to wake-up and when, knowing the algorithm code and the computation history but not the future randomness. The latter knows only the algorithm's code and has to fix its decision on what station to wake up and when before the execution starts (without knowing the random choices made by the stations). Clearly, the adaptive adversary is stronger than the oblivious one, in the sense that it could mimic any strategy of the oblivious adversary and additionally use online strategies,
\gdm{based on the knowledge of the computation history}, 
that may deteriorate the algorithm's performance even further.}

\paragraph{Timing.}
Although the communication is in synchronous rounds (\textit{i.e.}, the clocks of all the stations tick at the same rate) 
there is \textit{no global clock} and \textit{no system-based synchronization}: each station starts its local clock in the round in which it 
wakes up, without knowing anything about the round numbers ticked by the other clocks.
	We conventionally assume that a station is activated \gs{in} round 0 of its local clock and
	can start transmitting since round 1.

It is interesting to note that in the static model there is no distinction between the model with a global clock and that without it. Indeed, one can assume that a global clock is always available in that model: 
all the stations are activated simultaneously and therefore their clocks,
starting at the same time, will always tick the same round numbers. 
In this sense, the dynamic model considered in this work is more general and challenging
than the static one.

\paragraph{Metrics.}

In this paper we measure the efficiency of the algorithms both in terms
of \textit{time complexity} and \textit{energy consumption}. 
In a dynamic scenario when each station can be woken up at any time,
the activation times can be arbitrarily distant from each other,
therefore there is no straight way to correctly judge the time efficiency of an algorithm 
such as simply counting rounds from start to end (as it would be in the static model, where all stations are activated at once).
A natural criterion is to consider the \textit{latency} of the algorithm.
First we define the latency of a station as the
number of rounds necessary for the station to transmit successfully, 
measured since its activation time. Then, the maximum of these values, 
calculated among all $k$ stations, provides the latency of the algorithm.
\gdm{
Concerning the energy consumption, the algorithm's efficiency 
will be evaluated in terms of total number of transmissions (broadcast attempts)
performed by all stations executing the protocol. 
Both latency and energy cost will be analyzed 
against a worst-case adaptive adversary: the upper bounds will hold against {\em any} strategy of the (online) adaptive adversary. Our lower bound will hold even for a weaker oblivious adversary: the formula holds even if the adversary fixes its worst-case wake-up pattern {\em prior} to the execution.}

All our asymptotic \gia{upper} bounds are to be understood as high probability bounds, that is, 
they hold \textit{with high probability} (in short: \textit{whp}).
We say that an event for an algorithm holds whp,
when for a predefined parameter $\eta > 0$,
the parameters of the algorithm can be chosen,
so that for any contention size $k$ the event holds with probability at least
$1 - 1/k^{\eta}$.
In the intermediate steps of analysis, 
we will sometimes need to use the notion of whp 
not only with respect to the pre-assumed parameter $\eta$;
in such a case we will say more specifically 
that an event occurs ``whp $1-1/k^{\lambda}$'', 
for some $\lambda>0$.
Parameter $\lambda$ will typically be slightly higher than $\eta$, so
that at the end we could get the final result with the sought probability 
at least $1 - 1/k^{\eta}$.

\subsection{Previous work and our contribution}

Contention resolution on a shared channel
has a very long and rich history, including 
communication tasks, scheduling, fault-tolerance, security, energy, game-theoretical and many other aspects, 

The first theoretical papers on channel contention resolution
date back to the 70's and considered mainly solutions
\gia{either for the static scenario or the dynamic scenario restricted
to when the activation times follow some known 
probability distribution}. 
\gia{The common assumption is that the number of stations 
connected to the shared medium is very large with respect to the
actual number of stations that can be involved in the
contention.
In this case the simple time-division multiple access (TDMA), 
which assigns a different round to each of the potential transmitters, 
would become very inefficient.}

\gia{Abramson \cite{Abr70} introduced the first random-access 
technique, called the Aloha system, that, contrary to TDMA, instead of avoiding collisions, allows retransmission of the data packets when collisions occur.
Soon after, Roberts \cite{Rob72} designed a method to divide the continuous time into discrete time slots by allowing the stations to agree on slot boundaries (slotted Aloha system). These earliest works focused on
queueing models in which each station maintains a queue of data packets to
be sent that arrive according to independent Poisson processes.
The basic idea was ingeniously simple:
when a new packet arrives at a station, it is immediately transmitted and,
if a collision is detected, it is retransmitted at a randomly 
selected future time. The main issue with any Aloha-type approach
was the instability: eventually the system reaches 
a situation where the number of stations involved in retransmissions tends to infinity, while the throughput tends to zero \cite{BertGal}.

A novel category of protocol schemes, called splitting algorithms, 
introduced in the late 70's independently by Capetanakis \cite{Cap}, Hayes \cite{Hay}, and Tsybakov and Mikhailov \cite{TM},
allowed to control the random retransmissions in such a way 
to avoid many chaotic situations encountered in the previous solutions.
These algorithms were the first to be based on a 
\textit{collision-resolution} approach
aiming at recursively ``resolving'' collisions whenever they occur. 
The idea was to split the set 
of colliding stations into subsets, only one of which transmits 
in the subsequent time slot, while the other stations defer.
If another collision occurs, a further splitting into subsets is performed.
If in a time slot no station transmits or if exactly one station transmits in
it, the splitting stops. 

Alongside the queueing models discussed above, 
a common setting, which is also the one adopted in the present paper,
assumed that each new packet
arrives at a new station, rather than at a backlogged one.
In this context, the problem is to let each station to transmit a single message rather than a queue of messages. This is very realistic whenever
the number of stations is much larger than the arrival rate
(which is the scenario under which TDMA is very inefficient),
so that new arrivals at backlogged stations are negligible.
Moreover, from a theoretical point of view, this assumption captures the real
essence of multiaccess communication in that it assures a good 
approximation of a large number of systems with arbitrary assumptions 
about buffering of newly arriving packets 
(see \textit{e.g.} Section 4.2 in \cite{BertGal}.)
All of the following results assume, as the present paper, 
this setting with a single packet per station.

Under this setting, analysis and refinements of the splitting algorithms 
produced efficient solutions for the static scenario, \textit{i.e.}
when $k$ stations are activated simultaneously. The first 
rigorous analysis was made by Massey \cite{Massey81} who showed that the original splitting algorithm can solve a contention among $k$ stations
in $2.8867 k$ time slots in expectation, provided that $k$ is known.
Greenberg, Flajolet and Ladner \cite{GFL} and Greenberg and Ladner \cite{GL} 
presented an algorithm working 
in $2.134 k + O(\log k)$ expected rounds without any 
\textit{a priori} knowledge of the number $n$ of contenders. 
This was called a \textit{hybrid} algorithm in that it first produces 
an estimate of $k$ (which gives rise to the $O(\log k)$ additive term) 
and then uses a refinement of the original splitting algorithm to finally
resolve the contention. All of the above solutions, being based on
splitting algorithms, are adaptive and require a collision detection 
mechanism.
The same asymptotic (optimal) bound for the static scenario
has been obtained even 
without collision detection and with a non-adaptive algorithm that also ignores any \gia{initial} knowledge about the contention size $k$
\cite{sawtooth1, sawtooth2,AMM13}.
}
\textit{
This shows that in the static model, 
\textit{i.e., }
when all the packets arrive at the same time, there is
no asymptotic difference in the time complexity between adaptiveness and non-adaptiveness, even in the absence of any \gia{\textit{a priori}} knowledge about channel contention.}
In the dynamic scenario, considered in this paper, Bender et al.~\cite{Bend-16}
  designed an adaptive algorithm \textit{with collision detection}
that, without any given bound on parameter $k$, exhibits constant throughput, linear latency and $O(\log \log^{*} k)$ expected
transmissions per station.
More recently Bender et al.~\cite{Bend-20}
proved that constant throughput and polylogarithmic transmissions 
can be achieved,
with high probability, by adaptive algorithms even without collision detection.

\begin{center}
\begin{table*}[t]
\footnotesize 

\begin{tabular}{| c | c | c | c | c | c | c |}
    \hline
   \multirow{2}{*}{\textbf{synchrony}}       &  
   \multirow{2}{*}{\textbf{C.D.}}  &  
   \multirow{2}{*}{\textbf{setting}}         & 
   \multicolumn{2}{c|}{\textbf{latency}}     &
   \multicolumn{2}{c|}{\textbf{energy}}    \\ 
    \cline{4-7} 
   & & & \textbf{u. b.}     & \textbf{l. b.}       &
         \textbf{u. b.}     & \textbf{l. b.}      \\   
    \hline
    static    & no & non-adaptive and $k$ unknown  & $O(k)$ \cite{sawtooth1,sawtooth2,AMM13}  & $\Omega(k)$ & $\Omega(k\log^2 k)$ \cite{sawtooth1,sawtooth2,AMM13} & $\Omega(k)$\\
    \hline
    \multirow{4}{*}{dynamic}
                                  & yes & adaptive and $k$ unknown & $O(k)$ \cite{Bend-16} &  $\Omega (k)$
                                   & $O(\log(\log^* k))$ \cite{Bend-16} & $\Omega(k)$\\
                                  & no  & non-adaptive and $k$ known   & $\boldsymbol{O(k)}$ & $\Omega (k)$
                                   & $\boldsymbol{O(k\log k)}$ & $\Omega(k)$\\                                  
                                  & no & non-adaptive and $k$ unknown &  
                                  $\boldsymbol{O\left(k\frac{\ln^2 k}{\ln\ln k}\right)}$  &                        
                                  $\boldsymbol{\Omega\left(k\frac{\log k}{(\log\log k)^2}\right)}$
                                  & $\boldsymbol{O(k\log^2 k)}$ & $\Omega(k)$\\
                                  & no  & adaptive and $k$ unknown   & $\boldsymbol{O(k)}$ & $\Omega (k)$
                                   & $\boldsymbol{O(k\log^2 k)}$ & $\Omega(k)$\\
    \hline
\end{tabular}
\caption{Currently best results on latency \gia{and energy} in randomized contention resolution. \textit{Collision detection} is abbreviated ``C.D.'', while 
``u. b.'' and ``l. b'' respectively stand for \textit{upper bound} and 
\textit{lower bound}. Results from this paper are shown in bold. \dk{A lower bound $\Omega(k)$ is inevitable,
\gia{both for latency and energy}, to have $k$ successful transmissions.}}
\label{table:results} 
\end{table*}
\end{center}

\paragraph{Our contribution.}  
We pursue the study on contention resolution further by considering 
both adaptive and non-adaptive protocols in the general 
setting when \gia{stations awaken} at arbitrary times and in the 
severe model without collision detection.  
The wake-up times are determined by a worst-case adversary.
Upon waking up, each station starts its protocol from the beginning and a global clock is not available
(an upper bound for the case with global clock has been given  
for the deterministic setting~\cite{DK}).
\gia{With respect to the preliminary version \cite{DeMarcoS17}, the present work provides the following significant additions. The analysis of the algorithms have been improved to work against a stronger adaptive adversary, while the lower bound has been strengthen to deal even 
if the activation times of the stations are scheduled by a much 
weaker oblivious adversary. 
We consider also transmission energy efficiency, in particular improving 
the expected total number of transmissions of our adaptive algorithm 
from $\Omega (k^2)$ to $O(k \log^2 k)$.
Finally, we formulate several challenging open directions in the field of
shared channel communication.}

Our results can be summarized as follows (see Table~\ref{table:results} for a comparison with the static model in the most severe settings, \textit{i.e.} non-adaptivity without collision 
detection and with $k$ unknown).

\begin{itemize}
\item
\dk{If a linear upper bound on $k$ is given \textit{a priori} to the stations or the protocol is adaptive,} 
then the contention 
resolution can be done optimally (in asymptotic sense); 
we provide two corresponding algorithms working whp with latency $O(k)$, \dk{c.f., Theorem~\ref{th:main} in Section~\ref{NA-k-known} and Theorem~\ref{t:withLeader} in Section~\ref{s:no-GC-no-k}, respectively.}
 
\item
If the protocol is non-adaptive and no linear upper bound is known on parameter $k$, then 
we show a time complexity separation with the previous cases by proving that 
there is 
\textit{no} non-adaptive randomized algorithm, 
without the knowledge of any linear upper bound on $k$, achieving
$o\left(k\frac{\log k}{(\log\log k)^2}\right)$ latency whp, \dk{c.f., Theorem~\ref{t:lower-gen} in Section~\ref{sec:lowerbound}.}
 
\item
We also give a non-adaptive algorithm, ignoring any 
linear upper bound on $k$, 
\textit{i.e.}, an universal random code, with 
latency 
$O(k\log^2 k)$ whp,
which almost match
our lower bounds for the same setting \dk{(see Theorem~\ref{t:full-1} in Section~\ref{s:no-GC-no-k-no-a-algorithm}).}
Additionally we show that if we allow the stations to switch-off upon an acknowledgment then
this algorithm/code guarantees even a better latency: 
$O\left(k\frac{\log^2 k}{\log\log k}\right)$ whp,
\dk{c.f., Theorem~\ref{t:full-2} in Section~\ref{s:no-GC-no-k-no-a-algorithm}.}
\end{itemize}
\darek{All our upper bounds (achieved by algorithms)  hold for an adaptive adversary, while the lower bound holds even for a weaker oblivious adversary. In view of our (almost) tight formulas, it implies that the power of the adversary does not (substantially) influence the complexity of the contention resolution problem.}

\textit{
Our contribution implies that, in contrast with what happens in the static model, in the dynamic counterpart
there is a separation, 
in terms of time complexity (latency), between non-adaptive algorithms ignoring $k$ and algorithms that 
either are adaptive or know parameter $k$.  
It also implies a separation between the static and dynamic models, in case of non-adaptive algorithms
without a good estimate of the number $k$ of contenders.
}
It is interesting to note that, compared with the Bender et al.~\cite{Bend-16} protocol, our adaptive algorithm exhibits 
the same optimal performance on latency  
even in the more severe setting without collision detection.

\gdm{
Finally, all of our algorithms, despite the 
absence of a collision detection mechanism, 
are also efficient in terms of the total number of transmissions performed by the stations during the execution.
Namely, the adaptive algorithm spends  $O(k \log^2 k)$ 
broadcast attempts, \dk{c.f., Theorem~\ref{thm:energy-adaptive},} while our non-adaptive solutions, 
with and without knowledge of $k$, have an energy cost respectively 
$O(k \log k)$ and $O(k \log^2 k)$ whp, 
\dk{see Theorems~\ref{energy1} and~\ref{thm:energy-non-adaptive-unknown}.}
}

Given the generality of the 
shared channel
as a \textit{symmetry breaking} model, we 
believe that our contribution and newly developed techniques could shed light on the 
complexity of other problems in distributed computing.

\paragraph{Our approach \dk{and technical contribution}.}
We build our approach on the following four findings.

The first and key finding is
that in order to utilize the channel
efficiently against some non-synchronized sequence of activation times, 
any ``universally efficient'' (\textit{i.e.}, efficient for any contention size) 
\textit{non-adaptive} contention resolution algorithm 
has to schedule probabilities of transmissions in the first $\Theta(k\log k/(\log\log k)^2)$ rounds
in such a way that they sum up to $\Omega(\log k/\log\log k)$ 
(see Lemma~\ref{l:lower-gen-1} in Section~\ref{sec:lowerbound}).
(Recall that we do not require from the algorithms that these probabilities 
are independent over the rounds, which make our lower bound on latency more general.)
This was not the case for previous analysis of simplified static scenarios 
(\textit{i.e.}, when stations start the protocol at the same time)
or simpler problems such as wake-up (\textit{i.e.}, waiting for the first
successful transmission).
We later show that some slightly more subtle selection of activation times
pumps-up the sum requirement to $\Omega(\log^2 k/(\log\log k)^2)$ over the first 
$\Theta(k\log k/(\log\log k)^2)$ rounds (see Lemma~\ref{l:lower-gen-5}).
\gia{This is sufficient
to prove the corresponding superlinear} lower bound on latency (see Theorem~\ref{t:lower-gen}).

Second, we found out that such an effect does not hold if we know some linear
upper bound on the number of contenders, since slowly increasing probabilities,
starting from the level of $1/k$ and ending at $(\log k)/k$,
guarantees the existence of many rounds at which the sum of 
probabilities 
of the alive stations is smaller than $1$, regardless of how the
activation times are located (cf., Lemmas~\ref{inLemma3} to~\ref{Lemma4}
in Section~\ref{NA-k-known}). 
This, together with the
property that the transmission probability of a single station is at least
$(\log k)/k$, implies that each station succeeds in linear time
with high probability (cf., Lemma~\ref{Fact2} and the final proof of 
Theorem~\ref{th:main}).
Note that in this case the sum of transmission probabilities of a station
during the first $\Theta(k)$ rounds is also $\Omega(\log k)$,
however the knowledge of $k$ (or its linear upper bound) allows the algorithm
to schedule them in such a way that the pumping-up technique
from the lower bound mentioned above (for non-adaptive solutions
\textit{without} knowledge of $k$) does not work.

Third, in Section~\ref{s:no-GC-no-k} we will show how
the adaptiveness can also moderate the effect of the pumped-up sums, 
as it allows to elect a leader quickly using known
solutions to the wake-up problem. The leader could then
be used to coordinate the transmissions of the synchronized stations. 
The main challenging part is to manage
such a coordination among different 
algorithmic components
(wake-up and uniform selection of a subset of already synchronized stations) on an asynchronous channel without accessing to a global clock and with no \textit{a priori} information about the number of participating stations. 
We overcome these obstacles by using a few types of modes and a small number of transmissions.

Finally, in the most severe setting of non-adaptive algorithms with
no a priori knowledge, sub-linearly decreasing probabilities 
allow to find a linear fraction of rounds with $O(1)$ sum of transmission
probabilities, cf., Lemmas~\ref{f:full-7} to~\ref{f:full-10} in 
Section~\ref{s:no-GC-no-k-no-a-algorithm}. 
This, together with the fact that a station itself transmits
with sub-linearly decreasing probability $(\log j)/j$ in round $j$,
guarantees successful transmissions with high probability in a slightly
overlinear time period of $O(k\log^2 k/(\log\log k))$ rounds.

\gia{
\subsection{Related works}
An interesting related line of research studies the contention resolution problem in presence of adversarial jamming. 
Awerbuch et al. \cite{Aw} as well as Richa et al. \cite{Richa1, Richa2, Richa3}
studied jamming in multiple access channels in an adversarial setting 
when jamming is bounded within any sufficiently large fraction of the time.
Anantharamu et al. \cite{Anantha} considered the setting with a 
dynamic arrivals of packets subject to particular injection rates and 
jamming rates. 
For an account of the literature on adversarial models the interested reader 
can consult Richa et al.~\cite{Richa}.
For arbitrary jamming models we refer the reader also to the works 
by Alistarh et al.~\cite{Alistarh} and Gilbert 
et al.~\cite{GilbGuerNew}. Energy efficiency approaches can be 
found in \cite{Gilbert}.
More recently, the setting where $d$ slots can be jammed by an adaptive adversary has received attention. In such a situation, when collision
detection is available,
very efficient algorithms, both in terms of throughput and energy, have
been presented in \cite{BenderJamm, BendGilbJamm}.
Interestingly, Bender et al. \cite{Bend-20} showed a separation
between the models with and without collision detection, proving
that no constant throughput algorithm can be achieved in presence of jamming
when collision detection is not available.

As for the energy consumption, some papers also measured the efficiency of contention resolution analyzing the number of channel accesses, including 
not only the  broadcast attempts, but also the time spent by the stations listening to the channel. In this context, for the model with collision detection, Bender et al. \cite{Bend-16} presented a randomized algorithm 
with expected constant throughput and only $O(\log(\log^* N))$ 
channel accesses in expectation.
}

\subsection{Structure of the paper.}
\gdm{
We start with Section ~\ref{sec:notation} presenting some conventions and notations that will be used throughout the paper. 
The next three sections will be devoted to non-adaptive algorithms.
Namely,
in Section~\ref{NA-k-known} we give our optimal non-adaptive solution 
with known contention size;
in Section~\ref{sec:lowerbound} we show our lower
bound for non-adaptive algorithms unaware of the contention size; \dk{and}
in Section~\ref{s:no-GC-no-k-no-a-algorithm}
we give an almost optimal non-adaptive algorithm 
for unknown contention size. 
Finally, in Section~\ref{s:no-GC-no-k} we present our optimal adaptive solution
for unknown contention size.
}

\gia{
\section{Technical preliminaries}\label{sec:notation}

\subsection{Conventions and notation}
}
The activity of every station is articulated into synchronous rounds counted by a local clock. 
Each station can measure the time only on the base of this local numbering, without having any information on the clocks of other stations. 
By convention, we assume that a station is activated at round 0 of its local clock and can start its transmitting
schedule from the next round: \textit{i.e.}, at each round $1,2,3,\ldots$ a station can decide 
the probability of transmission.\footnote{%
The probabilities of transmission are not necessarily independent --- we follow
the general definition of randomized algorithms being a random distribution
over deterministic algorithms. Therefore, in case of adaptive algorithms, the probabilities
may depend also on the history of the channel feedback.
} 

The example given in Figure~\ref{tab:clocks} shows the lack of synchrony among the clocks of stations with different wake-up times: the first round of $u_4$ corresponds to the third round of $u_2$ and $u_3$ (that were activated simultaneously) and to the seventh round of $u_1$. 
	
Although there is no global time accessible to the stations, in the analysis we will need a \textit{reference clock} (not visible to the stations) that allows us to consider the behaviour of all the stations involved in the computation at a given moment. 
For example, in Figure~\ref{tab:clocks} we could start a reference
clock synchronously with $u_1$'s clock and say that 
at time $5$ (of the reference clock) there are three active stations: $u_1$, $u_2$ and $u_3$.

\begin{figure}
\centering
\begin{tabular}{rcccccccccc}
\toprule
  $u_1$'s local rounds: & 0 & 1 & 2 & 3 & 4 & 5 & 6 & 7 & 8 & 9 ...\\
  $u_2$'s local rounds: &   &   &   &   & 0 & 1 & 2 & 3 & 4 & 5 ...\\
  $u_3$'s local rounds: &   &   &   &   & 0 & 1 & 2 & 3 & 4 & 5 ...\\
  $u_4$'s local rounds: &   &   &   &   &   &   & 0 & 1 & 2 & 3 ...\\
              &   &   &   $\vdots$ &   & & &  &   &   &  \\   
   \hline
\end{tabular}
\caption{Lack of synchrony among the clocks.}\label{tab:clocks}
\end{figure}

For any time $t$ of a given reference clock, we denote by $\hat{A}[t]$ the set of stations
activated until time $t$.
The transmission probability assigned by the protocol to a station $v\in \hat{A}[t]$ at time $t$ 
will be denoted by $q_v[t]$.
Some stations may however cease to be active during the protocol, 
therefore we use $A[t] \subseteq \hat{A}[t]$ to designate the set of stations that are \textit{still} active at time $t$.
  
Moreover, we define the 
\textit{sum of transmission probabilities} at time $t$ as follows:
\[
                 \sigma[t] = \sum_{u\in A[t]} q_u[t] \ .
\]
The notation introduced so far will be used both for adaptive and non-adaptive
algorithms. 

Let us now introduce some definitions that 
will be specific for \textit{non-adaptive} algorithms.
When dealing with non-adaptive algorithms, we will also need to express the 
transmission probability of a station at any round of its \textit{local clock}: 
$p(i)$ denotes the probability that an arbitrary station transmits at the $i$th round of its local clock.
More precisely, without loss of generality, we can assume that each station $v$ chooses randomly, in advance,
the sequence of rounds in which it schedules the transmissions.
Recall that the assignment of transmissions to rounds does not need to be independent over the rounds.

Notice that in $p(i)$ we do not need to specify the station which this probability refers to. 
Indeed, for non-adaptive algorithms, the probability that a station transmits at the $i$th round of its local clock
may depend only on the local round index $i$ (recall that the stations are anonymous).
We also define the
\textit{sum of transmission probabilities of an arbitrary station} up to local time~$i$:
\[
	s(i) = \sum_{j=1}^{i} p(j) \ .
\]

Let $v$ be an arbitrary station and $t_v$ be the round, with respect to a given reference clock, corresponding to the wake-up time of $v$. The $t$th round of the reference clock corresponds to the $(t-t_v)$th round of $v$'s local clock. Therefore, if $v$ is active in round $t$, its transmission probability in this round is $q_v[t]=p(t-t_v)$.

We define
\[
\hat{\sigma}[t]=\sum_{v\in \hat{A}[t]} {p(t-t_v)} \ .
\]
It is easy to see that $\hat{\sigma}[t]\geq \sigma[t]$, due to 
the ranges of sums defining these two values.

As might have been noticed, when the time appears as argument of a function 
(like $p(i)$, $s(i)$, $A[t]$, $q_u[t]$), to avoid confusion,
we have used parenthesis for local rounds and square brackets for rounds of a given reference clock.
We will continue using such a convention throughout the paper.

\gia{
\subsection{Chernoff bounds.}
Throughout the paper, we will use the following form of the
Chernoff bound  (see Eq. (4.2) and (4.5) in 
\cite{MitUpf}).
Let $X_1 ,\ldots, X_n$ be independent Poisson trials such that 
$\Pr(X_i = 1) = p_i$ .
Let $X = \sum_{i=1}^{n} X_i$ and $\mu = \E[X]$. Then, for $0 < \delta < 1$,
the following bounds hold:
\begin{eqnarray*}
 \Pr (X \ge (1+\delta)\mu ) &\le& e^{-\frac{\delta^2\mu}{3}}\label{chernoff1}; \\
 \Pr (X \le (1-\delta)\mu ) &\le& e^{-\frac{\delta^2\mu}{2}}\label{chernoff2}.
\end{eqnarray*}
}
\section{A non-adaptive algorithm for known contention}\label{NA-k-known}

In this section we give a non-adaptive algorithm achieving latency $O(k)$
in the case when the number of contenders $k$, or a linear upper bound on $k$,
is given to the stations as a part of the input.
Recall that the stations do not have any ID and are allowed to wake up arbitrarily. 
Apart from being non-adaptive, an additional advantage of our algorithm is that it is uniform,
\textit{i.e.}, the transmissions are independent over rounds.
\gdm{Moreover, it can deal successfully against an adaptive adversary.
}


The algorithm is formally described as follows. Any station, starting from the
time at which it is activated, executes the following protocol \NSU\
(see Algorithm \ref{alg:suniform}).
It takes two parameters in input: the number $k$ of stations and a 
\gia{constant $c$. Such a constant determines} the probability of success of the algorithm: 
for any fixed parameter $\eta > 0$, we can choose an input value 
$c$ such that the algorithm succeeds with probability at least $1 - 1/k^{\eta}$
for any contention size $k$ (recall the definition of high probability given in the Introduction).
For every integer $0 \le l \le \log\log k$, let 
\[
\varphi (l) = \left\{
\begin{tabular}{cc}
  $\frac{k}{2^l}$, & if $l < \log\log k$;\\
  $k$, & if $l = \log\log k$.\\
\end{tabular}\right.
\]
\begin{algorithm}[ht]
	\caption{\NSU}
	\label{alg:suniform}
	\begin{algorithmic}[1]
    \For{$l=0,1,2,\ldots,\log\log k$}
        \For{$c\varphi(l)$ rounds} transmit with probability $\frac{2^l}{2k}$
        \EndFor
    \EndFor
    \end{algorithmic}
\end{algorithm}

Figure~\ref{tab:protocol1} shows the sequence of transmission probabilities for two stations, activated in different rounds, as it results from the execution of the first three 
iterations of the inner for-loop. Notice how the lack of synchronization between the two stations causes that in many rounds they transmit with different probabilities.

\begin{figure}
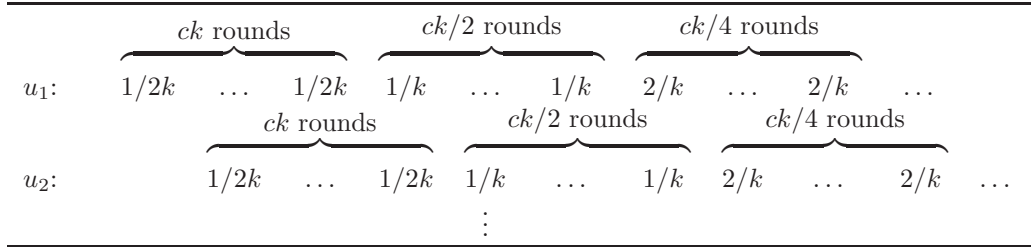

\centering
\begin{tabular}{rcccccccccccc}
\toprule
  & & \multicolumn{3}{c}{$\overbrace{\hspace{3cm}}^{\mbox{$ck$ rounds}}$} &  
  \multicolumn{3}{c}{$\overbrace{\hspace{3cm}}^{\mbox{$ck/2$ rounds}}$} &
  \multicolumn{3}{c}{$\overbrace{\hspace{3cm}}^{\mbox{$ck/4$ rounds}}$} & \\
 $u_1$: &  & $1/2k$ & $\ldots$ & $1/2k$ & $1/k$ & $\ldots$ & $1/k$ & $2/k$ & $\ldots$ & $2/k$ & \ldots\\
  & & & \multicolumn{3}{c}{$\overbrace{\hspace{3cm}}^{\mbox{$ck$ rounds}}$} &  
  \multicolumn{3}{c}{$\overbrace{\hspace{3cm}}^{\mbox{$ck/2$ rounds}}$} &
  \multicolumn{3}{c}{$\overbrace{\hspace{3cm}}^{\mbox{$ck/4$ rounds}}$} \\
 $u_2$: & & & $1/2k$ & $\ldots$ & $1/2k$ & $1/k$ & $\ldots$ & $1/k$ & $2/k$ & $\ldots$ & $2/k$ & \ldots\\ 
 
          &   &   &   &   &   & \vdots  &   &   &   &   &  \\  
          \hline 
\end{tabular}
\caption{Two stations executing the first 3 iterations of the inner for-loop of Protocol \NSU .}\label{tab:protocol1}
\end{figure}

In this section we want to prove the following theorem.

\begin{theorem}\label{th:main}
All stations executing protocol \NSU, for a sufficiently large constant $c$, will transmit successfully within $O(k)$ time rounds whp.
\gdm{The result holds even against an adaptive adversary.}
\end{theorem}

The linear upper bound on the latency of the protocol follows from 
an easy inspection of the pseudo-code:
the total number of rounds is 
	less than $c(k + k/2 + k/4 +\cdots + k/{\log k} + k) < 3c\cdot k$.
Therefore, we can state the following fact.

\begin{fact}\label{Fact1}
	For any given integer constant $c>0$, the number of rounds needed for any station to execute protocol \NSU\ is less than $3c\cdot k = O(k)$.
\end{fact}

Hence, from now on we may focus only on proving the correctness of the protocol.
As stated by Fact~\ref{Fact1},
each station is able to transmit only for at most (actually less than) $3c\cdot k$ rounds following its activation.
Therefore, there are altogether at most $3c\cdot k^2$ rounds in which the transmissions from the $k$ stations can occur. 
Of course, these rounds are not necessarily consecutive, as they
depend on the wake-up times of the stations which can be arbitrarily distant in time. 
This set of rounds can be partitioned into disjoint intervals
consisting of subsequent rounds at which there are some active stations that are executing the protocol and, therefore, can choose to transmit.
These intervals are called \textit{time frames} (there can be many of them).

Let's start with the following observation which is a consequence of the fact that there are altogether at most $3c\cdot k^2$ rounds in which a transmission can occur.

\begin{fact}\label{Fact1a}
For any choice of the activation times, any time frame lasts at most 
$3c\cdot k^2$ rounds. 
\end{fact}

In the following analysis we will concentrate on an arbitrary time frame $F$ and use a reference clock starting with it: the $t$th round of the reference clock is the $t$th round of $F$.
\gia{In other words, the reference clock is an imaginary global clock 
(unknown to the stations) starting with the activation time of the first
station(s).}
We will show that any station $v$ executing the protocol during such a time frame, will transmit successfully before
the end of the execution.  

\gia{The crux of the analysis will be to show that the rounds of a time frame
can be partitioned into $\log\log k$ disjoint intervals 
such that, as a result of successful transmissions, 
the number of active stations is halved in each of these intervals, independently of how their wake up times are chosen by the adversary 
(Lemma \ref{inLemma3}). 
Such a progressive reduction of the active stations, besides being beneficial 
in itself, in that it brings us closer to the final goal, it also remarkably contributes to lighten the contention among the remaining active stations,
so that it is possible to show that in each round $t$ of a time frame,
the sum of transmission probabilities $\sigma[t]$ is less than 1 whp (Lemma \ref{Lemma4}).
This is a very favorable situation 
that will be finally exploited by any station
that hasn't been able to transmit successfully until the last iteration,
\textit{i.e.}, when it transmits for $ck$ rounds with probability 
$\log k/(2k)$.
Namely, in the final proof of Theorem~\ref{th:main} 
on page \pageref{proofTh} we will show that having a small sum of 
transmission probabilities in every round of the time frame,
implies that any station -- if still alive -- transmits 
successfully with probability $\Omega(\log k/k)$ in each of the rounds
of the last iteration. This probability combined with the size $ck$ of 
the iteration assures that the station transmits successfully, whp, in one of these final rounds, provided it didn't do so previously.
}

We start with the next lemma stating that in each round $t\in F$ such that $\sigma[t] < 1$, we have a favorable 
probability that a given station $v$ transmits successfully.
For any round $t\in F$, let us define $\cE[t]$ as the event that $\sigma[t] < 1$. 

\begin{lemma}\label{Fact2}
If the event $\cE[t]$ holds, then the probability that a given station $v$ transmits successfully in round $t$ 
is larger than $q_{v}[t]/4$.
\end{lemma}

\begin{proof}
For every round $t$ of our protocol, $q_{v}[t]\le 1/2$,
\gia{from which we get $(1-q_{w}[t])^{\frac{1}{q_w[t]}} \ge 1/4$.}
Hence, the probability stated in the lemma becomes 
\begin{eqnarray*}
    q_{v}[t] \cdot \prod_{w\neq v} (1-q_{w}[t]) &=&    q_{v}[t] \cdot \prod_{w\neq v} (1-q_{w}[t])^{\frac{1}{q_w[t]}q_{w}[t]}\\ 
                                                &\geq& q_{v}[t] \cdot \prod_{w\neq v} (1/4)^{q_{w}[t]} \\
                                                &>&    q_{v}[t]\cdot (1/4)^{\sigma[t]} \\
                                                &>& q_{v}[t]/4 \ ,
\end{eqnarray*}
where the last inequality follows from the hypothesis that $\cE[t]$ holds.
\end{proof}

Our next goal will be to show that this favorable situation 
recurs whp for the whole execution 
of the algorithm.
Indeed, in Lemma~\ref{Lemma4} we will show that the events $\cE[t]$ simultaneously occur
whp for all rounds $t$ of a given time frame. Before being able to prove such a claim, 
we need three preparatory lemmas.

In the first of these lemmas we consider a scenario in which the sum of transmission probabilities 
$\sigma[\iota]$ is less than $1$ for all rounds $\iota< t$, 
\textit{i.e.},\ $\cE[\iota]$ occurs in all rounds $\iota$ preceding 
round $t$.

We show that for any $j\in(0, \log\log k]$,
during the $c\varphi(j)$ rounds preceding round $t$,
at least half of the stations that at time $t$ are assigned to some iteration $\lambda \ge j$ of the outer for-loop
of the protocol,
will transmit successfully whp.
Therefore they switch-off before round $t$.

Let $U$, $|U| \le k$, be the set of all the stations executing the algorithm in a given frame.
For any integer $l\in [0,\log\log k]$ and time round $t\ge 0$, we define $T^l[t]\subseteq U$ 
to be the subset of stations $v$ (not necessarily still active) 
that at round $t$ have assigned a 
transmission probability $q_v[t] = 2^l/(2k)$. In other words,
it includes all the stations that have been activated in some round 
before $t$
and would transmit with probability $2^l/(2k)$ if still active
(\textit{i.e.}, if not switched-off due to a successful transmission) at round $t$.

For any predefined parameter $\eta > 0$, let $c$ be chosen 
sufficiently large so that $\eta \le (c-8)^2/(32c) + 4$.

\begin{lemma}\label{inLemma3}
Let $0 < j \le \log\log k$ be an integer. 
Let $F$ be any time frame and $\tau, \tau'\in F$ be two time rounds such that $\tau' = \tau - c\varphi(j)$. Assume that
the events $\cE[1],\cE[2],\ldots,\cE[\tau-1]$ hold. Then,
if $T'\subseteq\bigcup_{\lambda\geq j}T^{\lambda}[\tau]$ and $|A[\tau'] \cap T' |\leq x$ for some $x>0$, 
then $|A[\tau] \cap T'|\leq \max\{x/2,\sqrt{k}\}$ 
with probability at least $1-k^{-\eta-4}$. 
\end{lemma}

\begin{proof}
Let us consider rounds $\iota\in[\tau',\tau-1]$. Notice that since $j>0$, all the stations in $T'$
must have joined the protocol before round $\tau'$. 
Thus during the interval $[\tau',\tau-1]$ the number of stations in $A[\iota] \cap T'$
cannot increase (it can only decrease as the stations switch-off after 
successful transmissions).

If $x\leq\sqrt{k}$, then
$|A[\tau] \cap T'| \leq |A[\tau'] \cap T'|\leq x \leq \sqrt{k}$ and there is nothing
to prove. So hereafter we assume that $x>\sqrt{k}$ and our aim will be to prove that 
$|A[\tau] \cap T'| \leq x/2$ with probability at least $1-k^{-\eta-4}$.

\gdm{
For any round $\iota$, $\tau' \le \iota \le \tau-1$, 
let $p_\iota$ be the probability of having a successful transmission at round $\iota$.
To any execution of the algorithm, we can attribute the following 0-1 sequence $\rho$, 
of length $c\varphi(j)$, where the $(\iota-\tau'+1)$th bit of $\rho$ corresponds
to round $\iota$, for $\iota = \tau',\tau'+1,\ldots,\tau-1$.
The $(\iota-\tau'+1)$th bit of $\rho$ is defined as follows.
}

\begin{itemize}
\item
If $|A[\iota] \cap T'|> x/2$ and there is no successful transmission from 
$A[\iota] \cap T'$ at time $\iota$, 
then we have a 0 on the position corresponding to $\iota$.
\item
\gdm{If $|A[\iota] \cap T'|> x/2$ and there is a successful transmission from 
$A[\iota]  \cap T'$ at time~$\iota$, then we assign 1 at the corresponding 
position with probability $(1/p_\iota) x2^j/16k$ and 0 otherwise.
Observe that this probability is well defined as $p_\iota > x2^j/16k$.
Indeed, there are more than
$x/2$ stations in $A[\iota]  \cap T'$ and each of them transmits successfully 
with probability larger than $q_v(\iota)/4 = 2^j/8k$ by Lemma~\ref{Fact2}.
}
\item
Finally, if $|A[\iota]  \cap T'|\le x/2$ then the position is set by tossing 
an asymmetric coin where 1 is output 
with probability $x2^j/16k$ and 0 otherwise.
\end{itemize}

\gdm{
We can now observe that for every $\tau' \le \iota \le \tau-1$, the 
probability of having a 1 in the corresponding $(\iota-\tau'+1)$th bit
of $\rho$ is ${x2^j}/{16k}$ independently of the values of the other bits,
even against an adaptive adversary. 
This is evident in case $|A[\iota]  \cap T'|\le x/2$.
In case $|A[\iota]  \cap T'| > x/2$ the probability of having a 1 
is $p_\iota (1/p_\iota){x2^j}/{16k} = {x2^j}/{16k}$, which turns out to be
independent of the probability $p_\iota$ of having a successful transmission 
at round $\iota$.
This in particular implies that although the probability of having a successful
transmission can be affected by the activation of new stations, the behaviour
of these newcomers cannot influence the probability of having a 1 or a 0 in $\rho$.
}

Let $X$ be the random variable defined as the number of ones in $\rho$.
We can now see that if $X\geq x/2$, then $|A[\tau]  \cap T'|\leq x/2$. 
	Assuming that $|A[\tau']  \cap T'| \leq x$, 
	the ones in $\rho$ produced when $|A[\iota] \cap T'|> x/2$ correspond to successful
	transmissions. Consequently, there are at most $x/2$ stations that 
	remain active in the interval $[\tau',\tau-1]$.
Recalling that by hypothesis 
$|A[\tau']  \cap T'| \leq x$, 
it follows that $|A[\tau]  \cap T'| \leq x-x/2 = x/2$.

Hence, the probability that $|A[t]  \cap T'|\leq x/2$ is at least the probability that $X\geq x/2$.
An estimate of the latter probability will complete the proof.
The random variable $X$ expresses the number of successes in 
$c\varphi(j) \ge ck/2^j$
mutually independent experiments,
in each of which the probability of success is $x2^j/16k$. Therefore $\E(X) \geq cx/16$.
\gia{Notice that for this inequality to hold, it is sufficient that $c\varphi(j) = ck/2^j$ for every $j=0,1, \ldots, \log\log k$, but recall that 
the algorithm uses $c\varphi(j) = ck$ rounds
in its last iteration, corresponding to $j = \log\log k$. 
}
Hence, for $c$ sufficiently large,
\begin{eqnarray*}
\Pr(X < x/2) &\le& \Pr(X \le \E(X)\cdot 8/c) \\
               &\le& \exp[-\E(X) (1-8/c)^2/2] \mbox{ (by the Chernoff bound)} \\
             &\le&   \exp[-\sqrt{k}(c-8)^2/(32c)] \mbox{ (recalling that $x > \sqrt{k}$)} \\
                 &<& k^{-(c-8)^2/(32c)} = k^{-\eta - 4} \ .
\end{eqnarray*}
\end{proof}

Our next 
milestone will be
to use Lemma~\ref{inLemma3} to show that when the events 
$\cE[1], \cE[2], \ldots, \cE[t-1]$ hold, then $\cE[t]$ also holds whp. 
This will be done in the two following lemmas. 
First, in Lemma~\ref{miss}, we define the more complex events $\cA^l[t]$ and prove
that their intersection implies event $\cE[t]$.
Second, in Lemma~\ref{conditional}, we apply this fact to upper bound the sought conditional probability.

Let $t\ge 0$ be a round index and $0 \le l \le \log\log k$.
Event $\cA^l[t]$ is the event that holds if and only if $|A[t] \cap T^l[t]| \leq \max\{|T^l[t]|/2^l,\sqrt{k}\}$.
Let $\cA[t]$ be the event that holds if and only if all events $\cA^l[t]$, 
for $l=0,1 \ldots, \log\log k$, hold.

\begin{lemma}\label{miss}
	Let $F$ be any time frame and $t \in F$.
If event $\cA[t]$ holds then $\cE[t]$ holds.
\end{lemma}

\begin{proof}
Fix a round index $t$ and assume that for $l=0,1, \ldots, \log\log k$, events $\cA^l[t]$ hold.
We have to show that $\sigma[t] < 1$. Indeed,
\begin{eqnarray*}
    \sigma[t] &=& \sum_{l=0}^{\log\log k} |A[t] \cap T^l[t]|\cdot 2^l/2k  \\
              &\leq& \sum_{l=0}^{\log\log k} |T^l[t]|/2k + \sum_{l=0}^{\log\log k} \sqrt{k}\cdot 2^l/2k 
                                    \mbox{ (by the hypothesis that $\cA[t]$ holds)} \\
          &<& |U|/2k + 1/2 \le 1.
\end{eqnarray*}
\end{proof}

For any event $\cA$ we denote by $\overline{\cA}$ the negation of $\cA$.


\begin{lemma}\label{conditional}
	Let $F$ be any time frame and $t \in F$. We have
\[
\Pr\left(\overline{\cE[t]}\given[\Big] \cE[1] \wedge \cE[2]\wedge\cdots \cE[t-1]\right) < k^{-\eta-3}(\log\log k+1).
\]
\end{lemma}

\begin{proof}
Let us fix any round $t$.
By Lemma~\ref{miss}, we have:
\begin{eqnarray}
  \Pr\left(\overline{\cE[t]}\given[\Big] \cE[1] \wedge \cE[2]\wedge\cdots \cE[t-1]\right)
       \leq
  \sum_{l=0}^{\log\log k} \Pr\left(\overline{\cA^l[t]} \given[\Big] \cE[1] \wedge \cdots \wedge \cE[t-1]\right).
\end{eqnarray}

In order to complete the proof it will suffice to show that 
$\Pr\left(\overline{\cA^l[t]}\given[\Big]\cE[1] \wedge \cdots \wedge \cE[t-1]\right)$ is smaller than $k^{-\eta-3}$ for every $l= 0,1, \ldots, \log\log k$.

For $l=0$, $\cA^l[t]$ reduces to 
$|A[t] \cap T^l[t]|\le \max\{|T^l[t]|,\sqrt{k}\}$ which trivially holds. Therefore, the probability of 
$\left(\overline{\cA^{0}[t]} \given[\Big] \cE[1] \wedge \cdots \wedge \cE[t-1]\right)$ is zero.

Let us now fix $l\in[1,\log\log k]$. We define the sequence of rounds $t_0,t_1,\ldots,t_l$
so that $t_j - t_{j-1} = c\varphi(j)$ and $t_l=t$
for $j = 1,\ldots,l$.
Note that $T^l[t]\subseteq\bigcup_{\lambda\geq j}T^{\lambda}[t_j]$ for any $j\leq l$.
Therefore, repeatedly applying Lemma \ref{inLemma3} to intervals $[\tau',\tau]=[t_{j-1},t_j]$ for $j=l, l-1, \ldots, 1$, we can write:
\begin{eqnarray*}
 |A[t_l] \cap T^l[t]| &\le&  \max\{|A[t_{l-1}] \cap T^l[t]|/2, \sqrt{k}\} 
                                                  \mbox{ (with probability at least } 1-k^{-\eta-4})\\
                        &\le&  \max\{|A[t_{l-2}] \cap T^l[t]|/4, \sqrt{k}\} 
                                                  \mbox{ (with probability at least } 1-k^{-\eta-4})\\ 
                        &\vdots& \\
                        &\le& \max\{|A[t_0]\cap T^l[t]|/2^l,\sqrt{k}\}
                                   \mbox{ (with probability at least } 1-k^{-\eta-4}) \\
                        &\le& \max\{|T^l[t]|/2^l,\sqrt{k}\} \ .
\end{eqnarray*}
%
%
Thus, after applying a union bound to the above derivation, we get that
\[
      |A[t] \cap T^l[t]| \le \max\{|T^l(t)|/2^l,\sqrt{k}\}
\]
holds with probability at least $1 - k^{-\eta-4}\log\log k > 1 - k^{-\eta-3}$.
\end{proof}

\begin{lemma}\label{Lemma4}
Given any time frame $F$, 
all events $\cE[t]$, for every $t \in F$, simultaneously occur with probability larger than $1 - k^{-\eta}/2$.
\end{lemma}

\begin{proof}
Let $F$ be any time frame.
We want to prove that
$1-\Pr(\cE[1] \wedge \cE[2] \wedge\cdots \wedge \cE[|F|])$ is less than $k^{-\eta}/2$.

Hence,
\begin{eqnarray*}
 \lefteqn{1-\Pr(\cE[1] \wedge \cE[2] \wedge\cdots \wedge \cE[|F|])}\\ &=&
  \sum_{t=1}^{|F|} \Pr\left(\overline{\cE[t]} \given[\Big] \cE[1] \wedge \cE[2]\wedge\cdots \cE[t-1]\right)\Pr(\cE[1]\wedge\cE[2]\wedge\cdots \cE[t-1])\\
    &\le& \sum_{t=1}^{|F|} \Pr\left(\overline{\cE[t]} \given[\Big] \cE[1] \wedge \cE[2]\wedge\cdots \cE[t-1]\right)\\
    &\le& 3c\cdot k^2 \cdot \Pr\left(\overline{\cE[t]} \given[\Big] \cE[1] \wedge \cE[2]\wedge\cdots \cE[t-1]\right) 
                    \mbox{ (by Fact~\ref{Fact1a})}\\
    &<& 3c\cdot k^2 \cdot k^{-\eta-3}(\log\log k+1) \mbox{ (by Lemma~\ref{conditional})}\\
    &<& k^{-\eta}/2 \ , 
\end{eqnarray*}
for sufficiently large $k$.
\end{proof}

The last lemma tells us that, whp, $\sigma[t]<1$ holds for all rounds $t$ within a time frame. 
Due to this inequality, we are now able to show that in each of the last $ck$ rounds, station $v$ --- if still awake --- transmits successfully with probability $\Omega(\log k/k)$.
This, in turn, assures that $v$ transmits successfully, whp, in one of these final rounds (provided it didn't do so previously).
What follows is the proof of Theorem~\ref{th:main}, which 
is the main result of this section.

\medskip

\noindent
{\bf Proof of Theorem~\ref{th:main}}\label{proofTh}

Let us focus on some station $v$ active in some time frame $F$.
\gia{We first prove that this station transmits whp within $O(k)$ rounds and 
then we take the union bound over all contending stations.}
We will prove that $v$ manages to transmit successfully, whp, by the end of the execution of its protocol.
Let $R \subseteq F$ be the set of the last $ck$ rounds 
executed by $v$ 
corresponding to the final iteration of the inner for-loop. 
While it is active,  station $v$ transmits in any round $t\in R$ with probability 
$q_v[t] = p(t-t_v) = \log k/(2k)$.
In order to prove the theorem, we need to show that $v$ transmits successfully in one of these rounds whp.

The probability that station $v$ does not transmit successfully during the rounds in $R$ 
is upper bounded by the probability that the station does not transmit successfully in a round $t\in R$ when all events $\cE[t]$ for $t\in R$ hold or that there exists a round $t\in R$ such that $\cE[t]$ does not hold. 

By Lemma~\ref{Fact2}, in any round $t\in R$ such that event $\cE[t]$ holds, the active station $v$ transmits successfully 
with probability larger than $q_v[t]/4=\log k/(8k)$.
So, assuming that for every round $t\in R$ event $\cE[t]$ holds, station $v$ 
does not manage to transmit successfully with probability less than
$(1-\log k/(8k))^{ck} < e^{-c\log k/8}$, which is smaller than $k^{-\eta}/2$, if we take the constant $c$ sufficiently large.

By Lemma~\ref{Lemma4} the probability that there exists a round $t\in R$ such that event $\cE[t]$
does not hold is also smaller than $k^{-\eta}/2$. 
Hence, the probability that station $v$ does not transmit successfully is less than 
$k^{-\eta}/2 + k^{-\eta}/2 = k^{-\eta}$.
\gia{Now, taking the union bound over all contending stations, we get that
the probability that one station fails to transmit successfully is at most
$k^{-\eta +1}$, for any prefixed parameter $\eta > 0$. 
This also means that for any fixed $\eta - 1 >0$,
all stations transmit successfully with probability at least $1 - k^{-(\eta-1)}$,
that is whp.}
This concludes the proof of Theorem~\ref{th:main}. \qed

\medskip

Now we can conclude this section by showing that our algorithm is also energy efficient.
\gdm{In the following theorem we consider the number of broadcast attempts performed by
all contending stations.}

\gdm{
\begin{theorem}\label{energy1}
The total number of broadcast attempts for an execution of Protocol \NSU\ is $O(k\log k)$ whp.
The result holds even against an adaptive adversary.
\end{theorem}
\begin{proof}
For  $l<\log\log k$, during the $l$-th execution of the main loop, the expected number of transmissions per station is
$c\varphi(l)\cdot 2^l/(2k)=c/2$. While for $l=\log\log k$ this expected number is 
$c\varphi(\log\log k)\cdot 2^{\log\log k}/(2k)=(c/2)\log k$. Overall, we have $O(\log k)$ expected transmissions per station for the whole execution of the protocol.

Since the transmissions are independent over rounds,
by using the Chernoff bound 
\darek{to estimate from above the number of transmissions of a single station, and then the union bound over all stations, 
we get $O(\log k)$ transmissions for every station, whp.}
This means that in total our algorithm requires $O(k\log k)$ broadcast attempts whp.
\end{proof}
}

\section{A lower bound for non-adaptive algorithms}\label{sec:lowerbound}


The aim of this section is to prove the following theorem stating a lower bound
on non-adaptive algorithms without the knowledge of $k$.

\begin{theorem} 
\label{t:lower-gen}
There is no non-adaptive contention resolution algorithm, with no knowledge of contention size $k$,
with a latency of $o(k\log k/(\log\log k)^2)$ rounds whp.
\darek{The result holds even for a weaker oblivious adversary.}
\end{theorem}

For the sake of contradiction, in the following we will assume that
there is a non-adaptive algorithm $\cA$ which, without having any information on 
contention size $k$, 
solves the contention for any number $k$ of stations, with 
a latency 
of $o\left(k\frac{\log k}{(\log\log k)^2}\right)$ rounds whp.

In the following, we will use the term {\em instance} to denote 
a schedule assigning to each station the time at which it becomes 
activated.
%
Our aim will be to show 
existence of an instance of the problem such that 
whp algorithm $\cA$ does not make {\em any} successful transmission on this instance
within the first $\Omega\left(k\frac{\log k}{(\log\log k)^2}\right)$ rounds.
This of course contradicts the assumption on the latency of algorithm $\cA$.
We stress out the following interesting polarization that we have discovered
and make use of here, and which does not take place in the simpler setting with synchronization --- 
if one tries to solve contention resolution (\textit{i.e.}, 
to have {\em all} stations successful) 
\gia{with an algorithm too greedy} in terms of channel
utilization, 
\textit{it may not achieve even
a single successful transmission.}
\tbc{We would like to emphasize that the instance is determined in advance so that this lower bound is valid for oblivious adversary.}

\gia{Recall} that for algorithm $\cA$ we can define a sequence of probabilities $p(1),p(2),p(3),\ldots$
of transmissions of station's message in its local rounds $1,2,3,\ldots$ counting from its activation.
This sequence is the same for all executions and across all stations, due to $\cA$ being non-adaptive.
Namely, any station transmits its message with probability 
$p(1)$ in the first round after it has been activated.
If it has not transmitted successfully before, 
the transmission occurs with probability $p(i)$ for $i > 1$ in the $i$th round after its activation.
Note that the transmission in a round $i$ does not have to be necessarily an event independent on the transmissions in previous rounds; in this sense our proof
is more general than many other lower bounds in the literature of shared channel
which hold under the assumption of such an independence over rounds.

\tbc{
Without loss of generality we can assume $p(1) > 0$. Moreover, since the stations do not know parameter $k$, the contention size,
the probability $p(1)$ does not depend on $k$.}


%
We emphasize that, in view of our contradictory hypothesis, all the following 
definitions and results hold for a non-adaptive algorithm $\cA$ which, without knowing $k$, 
solves the Contention Resolution problem for {\em any number} $k$ of 
contending stations, with a time 
complexity of $o\left(k\frac{\log k}{(\log\log k)^2}\right)$ whp.

The following new definitions will be used.
Let $I(k)$ be an arbitrary instance in which exactly $k$ stations are activated.
Let $\tau(I(k))$ be the minimum time needed for algorithm $\cA$ to assure that
any activated station transmits successfully in instance $I(k)$ whp,
\textit{i.e.} at least $1-k^{-\eta}$ for any predetermined constant $\eta > 0$.
Let $\tau(k)=\max_{I(k)}\{\tau(I(k))\}$, where the maximum is taken over
all instances activating $k$ stations.
A straightforward consequence of our contradictory assumption is that 
$\tau(k) = o\left(k\frac{\log k}{(\log\log k)^2}\right)$.
We also define 
\[
   \varsigma (k) = s(\tau(k)) = \sum_{i\in[1,\tau(k)]} p(i) \ .
\]

Fixed any instance $I(k)$, we will use a reference clock starting 
at the first activation time of the instance. All the following 
rounds $t$ refer to this clock.  

The proof of the lower bound consists of two parts.
In the first part we show a dependance between 
the number of successful transmissions 
and the sum of transmission probabilities.
This is the task of the following lemma which  shows that 
if in an interval of $O(k^2)$ rounds the sum $\hsigma[t]$ of
transmission probabilities of all activated stations is $\Omega(\log k)$,
then whp no transmission will be successful in that interval. 
Note that, since $\tau(k) \le k^2$, this contradicts the assumption 
that algorithm $\cA$ has latency $\tau(k)$ whp.
The second part of the proof is devoted to constructing 
an instance of activation times such that the hypothesis of 
Lemma~\ref{l:lower-gen-2} holds, so to block whp successful transmissions
for $\tau(k)$ rounds.

%

\begin{lemma}
\label{l:lower-gen-2}
Fix an arbitrary instance.
Assume that for any round $t\in[1,T]$, with $T\le k^2$, it holds that 
$\hsigma[t] \geq \gamma\log k$ for some sufficiently large constant $\gamma>0$. Then,
no station transmits successfully by round $T$ with probability $1-1/k$.
\end{lemma}

\begin{proof}
We want to show that the probability of having at least one successful transmission
in the time interval $[1,T]$ is smaller than $1/k$.
Since we are interested in the probability of having \textit{at least one} successful transmission,
we can consider the simplified model in which the stations do not switch off when they transmit successfully.
Indeed, on any sequence of activation times for the $k$ stations,
the probability of having the first (\textit{i.e.}, at least one) successful transmission
on the original model is equivalent to the probability of having the first
transmission in the simplified model.
In the simplified model, the probability of successful transmission in any round $t$,
if no successful transmission happened before time $t$, is at most
\[
       \sum_{v\in \hA[t]} p(t-t_v)\prod_{{w\in \hA[t],} {w\not = v}} \left(1 - p(t-t_w) \right) \le \hsigma[t] e^{-\hsigma[t]+1}
\ ,
\]
which can be made smaller than $1/k^3$, for a sufficiently large $\gamma$.

By taking the union bound over all the $T\le k^2$ rounds, we get that the probability of no successful 
transmission in the simplified model (and thus in the original model) is at most $1/k$.
\end{proof}


The construction of the problem instance assuring the lower bound on the algorithm
will be based on the inequality $\varsigma(k)=\Omega(\log^2 k/(\log\log k)^2)$ 
proved in the key Lemma \ref{l:lower-gen-5},
transformed later into the lower bound on $\hsigma[t]$ in Lemma~\ref{l:lower-gen-6}.
In order to prove that key lemma, 
we first show a weaker inequality that $\varsigma(k)$
has to fulfill.

\begin{lemma}
\label{l:lower-gen-1}
$\varsigma(k) =\Omega(\log k/\log\log k)$.
\end{lemma}

\begin{proof}
To prove the lemma, we construct a random instance for algorithm $\cA$ 
such that if $\varsigma(k) \not = \Omega(\log k/\log\log k)$,
then 
the first station $v$ activated in this instance
transmits successfully, within $\tau(k)$ rounds, 
with probability smaller than $1-1/k^\zeta$,
for some constant $\zeta>0$.
Because of this, the \gia{contention} resolution problem 
would not be solved whp by algorithm $\cA$.

Let station $v$ be activated among the earliest possible group of stations, \textit{i.e.}, 
at round~1 of the instance. 
Station $v$ chooses at random the rounds in which it is going to transmit
(it can be assumed that this choice be done at the moment $v$ is activated, 
as the algorithm is non-adaptive).
The number of such rounds in the time period $[1, \tau(k)]$ 
is a random variable $X$ such that $E(X) = \varsigma(k)$. 
By Markov's inequality we have
\begin{equation}\label{eq:mark}
  \Pr\left(X < 2 \varsigma(k)  \right) > 1/2 \ .
\end{equation}
We now let the activation times of the other $k-1$ stations be distributed uniformly at random among the
$\tau(k) = O(k\log k)$ rounds. Each station transmits with probability
$p(1)$ at the first round it switches on. 
Therefore, at any round in which $v$ transmits, the 
probability that this transmission is not successful is at least 
$(k-1)\cdot p(1)\cdot(1/O(k\log k)) = \Omega(1/\log k)$.
\tbc{(Recall that, as explained in the introductory part of this section, $p(1)$ does not depend on~$k$.)}

Thus, the probability of no successful transmission for station $v$ during the interval 
$[1, \tau(k)]$ 
is
at least
\begin{eqnarray*}
    \Omega\left( \frac{1}{(\log k)^X} \right) \ .
\end{eqnarray*}
Hence,  
this probability (that $v$ does not transmit successfully) is negligible,
\textit{i.e.}, at most $1/k^\eta$ for any predetermined constant 
$\eta > 0$,
only when $X = \Omega(\log k/\log\log k)$.
By Equation~(\ref{eq:mark}), it follows that 
$\varsigma(k) = \Omega(\log k/\log\log k)$, or otherwise station $v$ does not transmit successfully whp within the interval $[1, \tau(k)]$.

\end{proof}

In the following lemma we show that the sum of probabilities in a round
could be indeed made $\Omega(\log n)$ over a period slightly shorter than the
ultimate length $\tau(k)$; we will extend it even further later on.

\begin{figure}[tp]
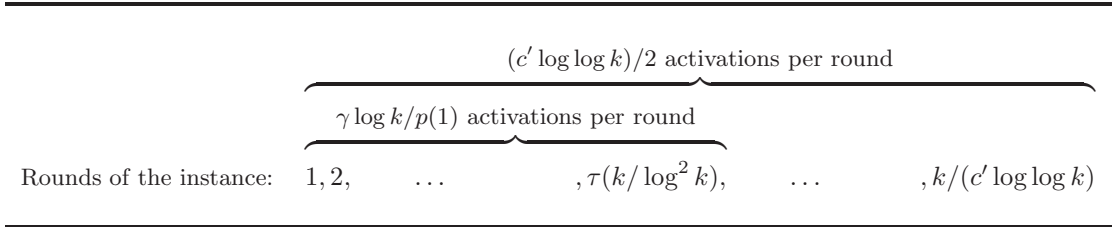


\centering
\begin{tabular}{rrrrrrrrrrr}
\toprule
\\
 & \multicolumn{9}{c}{$\overbrace{\hspace{10.5cm}}^{\mbox{\small $(c'\log\log k)/2$ activations per round}}$} \\ 
 & \multicolumn{5}{c}{$\overbrace{\hspace{5.6cm}}^{\mbox{\small $\gamma\log k/p(1)$ activations per round}}$} & & & & \\ 
{\small Rounds of the instance:} & $1,2,$ &  & $\ldots$ &  & $,\tau(k/\log^2 k),$ & & $\ldots$ & & $,k/(c'\log\log k)$ \\ \\
          \hline 
\end{tabular}

\caption{
Construction of the activation instance used in Lemma~\ref{l:lower-gen-3}.}
\label{tab:lemma3}
\end{figure}

From now on, we consider the constant $\gamma$ determined in Lemma~\ref{l:lower-gen-2}

\begin{lemma}
\label{l:lower-gen-3}
There is a constant $c_1>0$ and an integer $k_0$ such that for $k>k_0$ there is an instance $I(k)$ such that
$\hsigma[t] \ge \gamma\log k$ 
in rounds $t\in [1,\tau(k/(c_1\log k\log\log k))]$.
\end{lemma}

\begin{proof}
By the contradictory hypothesis on the latency of algorithm $\cA$, 
it follows that
$\tau(k) = o(k\log k/(\log\log k)^2)$.
This implies that for any
constant $c'>0$ there exists a constant $c_1>0$ such that
$\tau(k/(c_1\log k\log\log k))\leq k/(c'\log\log k)$ for $k$ sufficiently large.
In other words, for any $c'$ there is $c_1$ such that
\begin{equation}\label{inclusion}
  [1,\tau(k/(c_1\log k\log\log k))]\subseteq[1,k/(c'\log\log k)]
	\ .
\end{equation}
Also, we have 
\begin{equation}\label{second}
\tau(k/\log^2 k) = o(k/\log k)
\ .
\end{equation}
Let $\gamma$ be the constant determined in Lemma~\ref{l:lower-gen-2}.
Since, by Lemma~\ref{l:lower-gen-1}, $\varsigma(k) = \Omega(\log k/\log\log k)$,
there is a constant $c'$ 
such that for $k$ sufficiently large the following inequality holds:
\begin{equation}\label{vars}
  \varsigma(k/\log^2 k)\geq 2\gamma\log k/(c'\log\log k)
	\ .
\end{equation}

We can construct instance $I(k)$ for $k$ contending stations as follows
(see Figure~\ref{tab:lemma3} for a reference).
In each round $t\in[1,\tau(k/\log^2 k)]$ we switch on $\gamma\log k/p(1)$ stations.
Recalling that each station transmits with probability $p(1)$ 
in the first round after it has been activated,
we have $\hsigma[t] \geq\gamma\log k$ in all rounds $t\in[1,\tau(k/\log^2 k)]$. 
By Equation~(\ref{second}) it follows that, 
for sufficiently large $k$, 
it is sufficient to activate
at most $k/2$ stations in these rounds.

We continue the definition of the instance by switching on the remaining
stations as follows. 
The goal is now to guarantee $\hsigma[t] \geq\gamma\log k$ 
in all rounds $t\in[\tau(k/\log^2 k),k/(c'\log\log k)]$.
We switch on $(c'\log\log k)/2$ stations in each 
round $t\in [1,k/(c'\log\log k)]$ 
(note that this is possible since there are at least $k/2$ remaining stations).
Observe that in any round $t\in [\tau(k/\log^2 k), k/(c'\log\log k)]$, the following inequalities hold:
\[
\hsigma[t] \ge \varsigma(k/\log^2 k)\cdot c'\log\log k/2 \geq\gamma\log k
\ .
\] 
Indeed, the first inequality holds because for each considered round $t$
and for every integer 
$i\in [1,\tau(k/\log^2 k)]$
there are exactly
$(c'\log\log k)/2$ stations each contributing $p(i)$ to the sum $\hsigma[t]$,
totaling in $(c'\log\log k)/2$ times $\varsigma(k/\log^2 k)$;
the second inequality follows from Equation~(\ref{vars}).

This way we have showed that there exists an instance such that, 
for some constant $c'$,
the lower bound $\hsigma[t] \geq\gamma\log k$ holds in rounds $t\in[1,k/(c'\log\log k)]$.
Now, recalling Equation~(\ref{inclusion}), the lemma follows.
%
\end{proof}

Using the activation times from Lemma~\ref{l:lower-gen-3},
we can extend the statement of Lemma~\ref{l:lower-gen-1} by proving that
the transmission probabilities concentrate in some
suffix of the considered period $[1,\tau(k)]$.

\begin{lemma}
\label{l:lower-gen-4}
For some constant $c_2>0$, we have
\[
\varsigma(k) - \varsigma(k/(c_2\log k\log\log k)) = \Omega(\log k/\log\log k)
\ .
\]
\end{lemma}

\begin{proof}
Let $I(k/2)$ be the instance from Lemma \ref{l:lower-gen-3} for $k/2$ stations.
By Lemma~\ref{l:lower-gen-2}, for this instance, there are no successful transmissions in rounds $1,2,\ldots,\tau((k/2)/(c_1\log (k/2)\log\log (k/2)))$, whp. 
Note that for some constant $c_2$, we have $\tau((k/2)/(c_1\log (k/2)\log\log (k/2)))\ge \tau(k/(c_2\log k\log\log k))$.
Therefore, there are no successful transmissions in rounds 
$1,2,\ldots,\tau(k/(c_2\log k\log\log k))$, whp.
Consider a station $v$ that starts in round $1$ in instance $I(k/2)$.
It follows that in round $\tau(k/(c_2\log k\log\log k))$ the station is still 
running the protocol, whp. 

We then proceed analogously as in the proof of Lemma~\ref{l:lower-gen-1}.
We distribute randomly the remaining $k/2$ stations in the interval 
$[\tau(k/(c_2\log k\log\log k)),\tau(k)]$.
Station $v$ 
has 
some sequence of transmissions in this interval.
Each transmission is not successful with probability 
larger than
$(k/2)\cdot p(1)\cdot(1/\tau(k)) = \Omega(1/\log k)$.
Thus, in order to have a successful transmission with high probability,
station $v$ needs to transmit $\Omega(\log k/\log\log k)$ times
in the interval $[\tau(k/(c_2\log k\log\log k)),\tau(k)]$.

Therefore, the expected number of transmissions in the interval
$[\tau(k/(c_2\log k\log\log k)), \tau(k)]$,
which is $\varsigma(k) - \varsigma(k/(c_2\log k\log\log k))$, is 
$\Omega(\log k/\log\log k)$.
\end{proof}

Using Lemma~\ref{l:lower-gen-4} in a telescopic way, we can increase
the requirement on the sum of transmission probabilities of a single station
over the considered period $\tau(k)$.

\begin{lemma}
\label{l:lower-gen-5}
$\varsigma(k)=\Omega(\log^2 k/(\log\log k)^2)$.
\end{lemma}

\begin{proof}
We can write down a telescoping sum, where $c_2$ is the constant determined in Lemma \ref{l:lower-gen-4}:
\begin{eqnarray*}
\varsigma(k) &=&  \left(\varsigma(k) - \varsigma(k/(c_2\log k\log\log k))\right) +\\
             & &  \left(\varsigma(k/(c_2\log k\log\log k)) - \varsigma(k/(c_2\log k\log\log k)^2)\right) +\\
             & &  \left(\varsigma(k/(c_2\log k\log\log k)^2) - \varsigma(k/(c_2\log k\log\log k)^3)\right) 
                  + \ldots  
\end{eqnarray*}
The thesis follows by observing that this sum has
$\Omega(\log k/\log\log k)$ terms, and the first half of these terms
are $\Omega(\log k/\log\log k)$ by Lemma~\ref{l:lower-gen-4}. 
\end{proof}

We now transform the lower bound on transmission probabilities
of a station over the considered period $[1,\tau(k)]$, 
proved in Lemma~\ref{l:lower-gen-5},
into the lower bound on the sums of transmission probabilities in rounds.
Using Lemma \ref{l:lower-gen-5}, one can construct an instance
$J(k)$ for which no successful transmission is likely.

\begin{lemma}
\label{l:lower-gen-6}
\tbc{For any sufficiently large integer $k$ and some constant $c^*>0$,
there is an instance $J(k)$, which can be used by an oblivious adversary,
such that
the lower bound $\hsigma[t] \geq\gamma\log k$ holds
in all the rounds $t\in [1,c^* k\log k/(\log\log k)^2]$.}
\end{lemma}

\begin{proof}
We know, by the contradictory hypothesis, that $\tau(k/\log^2 k) = o(k/\log k)$.
By Lemma \ref{l:lower-gen-5}, there
is a constant $d$ such that, for sufficiently large $k$,
\begin{equation}\label{eq:sig}
  \varsigma(k/\log^2 k) \geq 2d\log^2 k/(\log\log k)^2
	\ .
\end{equation} 
The instance can be constructed as follows.

In each round $t\in [1,\tau(k/\log^2 k)]$ the oblivious adversary switches on
$\gamma\log k/p(1)$ stations.
This assures $\hsigma[t] \geq\gamma\log k$ in these rounds. 
One can use (no more than) $k/2$ stations in these rounds, 
for sufficiently large $k$.
Next, one distributes the activation times of the other $k/2$ stations.
For any constant $c > 0$, which we will fix later on,
the activation times of the other $k/2$ stations can be distributed
uniformly and independently at random in rounds 
$t\in [1,c k\log k/(\log\log k)^2]$. 
Recalling Equation~(\ref{eq:sig}), the average sum of probabilities $E(\hsigma[t])$
for these $k/2$ stations in rounds 
$t\in[\tau(k/\log^2 k),ck\log k/(\log\log k)^2]$
is at least
\[
\frac{k}{2} \cdot \frac{(\log\log k)^2}{ck\log k} \cdot \frac{2d\log^2 k}{(\log\log k)^2} = \frac{d}{c} \cdot \log k
\ .
\]

Letting $\delta = 1 - (\gamma c)/d$, we have
$(1-\delta) E(\hsigma[t]) = \gamma \log k$. Hence, by the Chernoff bound,
the probability that 
$\hsigma[t] < \gamma \log k$ in any round $t \in [\tau(k/\log^2 k),ck\log k/(\log\log k)^2]$ is at most
\[
   e^{-\delta^2 E(\hsigma[t])/2} = 
   e^{-\delta^2  \frac{d}{2c} \cdot \log k}=
   k ^{-(1 - \frac{\gamma c}{d})^2 \frac{d}{2c\ln 2}}.
\]
For a sufficiently small $c>0$, which we denote $c^*$,
the latter quantity can be made less than $k^{-3}$.
Applying the union bound
over all the rounds of the interval $[\tau(k/\log^2 k),ck\log k/(\log\log k)^2]$,
we get that the probability that there exists a 
round $t$ in that interval
such that $\hsigma[t] < \gamma \log k$ is
less than $k^{-1}$, for sufficiently large $k$.


Hence, drawing a random instance using this procedure we get the 
desired instance $J(k)$ 
with probability $1-1/k > 0$,
which proves that it exists.
\end{proof}

The proof of the main result is now straightforward.
\medskip

\noindent
{\bf Proof of Theorem~\ref{t:lower-gen}}.
The contradiction is obtained by applying Lemma~\ref{l:lower-gen-2} to the rounds
specified in Lemma~\ref{l:lower-gen-6}.
Since no station transmits in instance $J(k)$ whp $1-1/k$,
station $v$ cannot achieve latency $o(k\log k/(\log\log k)^2)$ whp.

\section{A non-adaptive algorithm for unknown contention}\label{s:no-GC-no-k-no-a-algorithm}

In this section we describe a non-adaptive protocol that resolves the conflicts without having any a priori knowledge on the number $k$ of contenders.
The algorithm can be formally described as follows. Starting from the time at which it is activated, any station executes the following protocol \NAdaptiveU\ (see Algorithm \ref{alg:non-adaptive-no-k}).
It takes a constant parameter $b>0$ in input that determines the probability of success of the algorithm (the larger $b$, the larger the probability that the algorithm succeeds).

\begin{algorithm}[t!]
	\caption{\NAdaptiveU}
	\label{alg:non-adaptive-no-k}
	\begin{algorithmic}[1]
    \For{$j=3,4,5,\ldots, \infty$}
    	\For{$b$ rounds}
               \State transmit with probability $\frac{\ln j}{j}$
	        \label{l:FA-2}
	\EndFor
    \EndFor
    \end{algorithmic}
\end{algorithm}


\begin{figure}[t!]
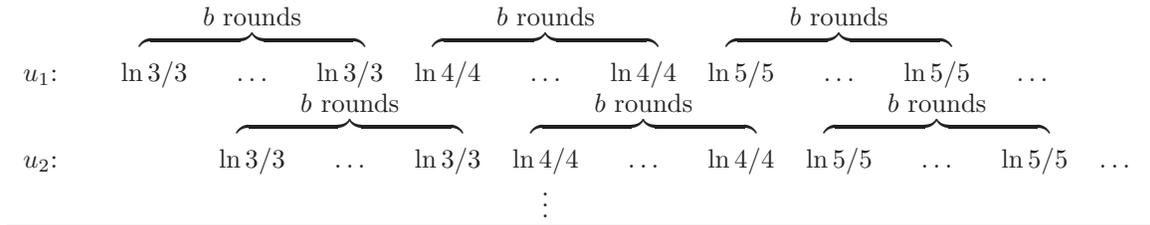

\centering

\begin{tabular}{rcccccccccccc}
\toprule
  & & \multicolumn{3}{c}{$\overbrace{\hspace{3cm}}^{\mbox{$b$ rounds}}$} &  
  \multicolumn{3}{c}{$\overbrace{\hspace{3cm}}^{\mbox{$b$ rounds}}$} &
  \multicolumn{3}{c}{$\overbrace{\hspace{3cm}}^{\mbox{$b$ rounds}}$} & \\
 $u_1$: &  & $\ln 3/3$ & $\ldots$ & $\ln 3/3$ & $\ln 4/4$ & $\ldots$ & $\ln 4/4$ & $\ln 5/5$ & $\ldots$ & $\ln 5/5$ & \ldots\\
  & & & \multicolumn{3}{c}{$\overbrace{\hspace{3cm}}^{\mbox{$b$ rounds}}$} &  
  \multicolumn{3}{c}{$\overbrace{\hspace{3cm}}^{\mbox{$b$ rounds}}$} &
  \multicolumn{3}{c}{$\overbrace{\hspace{3cm}}^{\mbox{$b$ rounds}}$} \\
 $u_2$: & & & $\ln 3/3$ & $\ldots$ & $\ln 3/3$ & $\ln 4/4$ & $\ldots$ & $\ln 4/4$ & $\ln 5/5$ & $\ldots$ & $\ln 5/5$ & \ldots\\ 
 
          &   &   &   &   &   & \vdots  &   &   &   &   &  \\  
          \hline 
\end{tabular}

\caption{Two stations executing the first 3 iterations of the inner for-loop of Protocol \NAdaptiveU .}\label{tab:protocol4}
\end{figure}

%

Figure~\ref{tab:protocol4} shows the sequence of transmission probabilities 
used by two stations
during the execution of the first three iterations of the inner for-loop. 
As can be seen, as a consequence of not being synchronized, 
in many rounds they transmit with different probabilities.


We start with proving a weaker bound on the complexity of protocol \NAdaptiveU\ that holds when acknowledgments are not allowed and
consequently a station does not switch-off after a successful transmission. This will give a first glimpse of our proof's technique. Later, in order to get our final result, we will refine the analysis by considering the acknowledgments and therefore exploiting the benefit of switching-off a station as soon as it successfully transmits.
\tbc{Both variants are successful against an adaptive adversary.}

\subsection{Analysis of the protocol without acknowledgments.}\label{no_ack}

In this subsection we will prove the following theorem.

\begin{theorem} 
\label{t:full-1}
All stations executing protocol \NAdaptiveU\, for a sufficiently large constant $b$, will transmit successfully
within  $O(k\ln^2 k)$ time rounds whp, even in the case where 
acknowledgments are not allowed.
\gdm{The result holds even against an adaptive adversary.}
\end{theorem}

All the following results refer to an execution of protocol \NAdaptiveU\ for an arbitrary value of its input parameter $b$.
We start with the following two technical facts.


\begin{fact}
\label{f:full-1}
For a sufficiently large $i$, we have $s(i) < b\ln^2 (i/b)$.
\end{fact}

\begin{proof}

We have,
\begin{eqnarray*}
      s(i) &\le& b\sum_{j=3}^{\lceil i/b \rceil} \frac{\ln j}{j}\\
           &\le& b\int_{j=2}^{\lceil i/b \rceil} \frac{\ln x}{x}\, dx 
               \;\;\; \text{ (for $i > 2b)$}\\
           &=&   b\cdot \frac{\ln^2(\lceil i/b \rceil) - \ln^2(2)}{2}\\
           &\le& b\ln^2( i/b ).
\end{eqnarray*}
%
\end{proof}

\begin{fact}
\label{f:full-2}
If $r = 4\cdot k\ln^2 k$, 
then $k\ln^2 r<r/2$ for a sufficiently large $k$.
\end{fact}

\begin{proof}

We have,
$k\ln^2 r = k\ln^2(4\cdot k\ln^2k) = k\ln^2 4 + k \ln^2 k + 2k\ln^2 \ln k <  2k \ln^2 k = r/2$; where the last inequality holds for a sufficiently large $k$.
\end{proof}

From now on we fix an arbitrary station $v$ with the aim of proving that such a station will transmit successfully within
$O(k\ln^2 k)$ time rounds whp. All the following rounds $t$ are referred to a reference clock 	
corresponding to $v$'s local clock, \textit{i.e.} starting at the wake-up time of $v$.
\tbc{Due to the adaptive adversary, the rounds $t:\hsigma[t]<1$, which are most favourable for successful transmissions of $v$ are not known in advance.
Nevertheless the next lemma guarantees they are half of the rounds in the aforementioned period.}

\begin{lemma}
\label{f:full-3}
Let $r = 4\cdot k\ln^2 k$. In at least $br/2$ rounds $t\in[1,br]$ we have  $\hsigma[t] < 1$.
\end{lemma}

\begin{proof}
Let $r = 4\cdot k\ln^2 k$. We have,

\begin{eqnarray*}
	\sum_{t\in [1,br]} \hsigma[t] &=& \sum_{t\in [1,br]} \sum_{w\in \hat{A}[t]} q_w[t] \\
	                              &=&  \sum_{w\in \hat{A}[t]} \sum_{t\in [1,br]} p(t-t_w) \\ 
	                              &\le& k\cdot s(br)\\
	                              &<& bk\ln^2 r  \mbox{ (by Fact~\ref{f:full-1})} \\
	                              &<& br/2 \mbox{ (by Fact~\ref{f:full-2})},
\end{eqnarray*}
where the last two inequalities hold for a sufficiently large $k$.
Thus, in at least $br/2$ rounds $t\in[1,br]$, the sum $\hsigma[t]$ must be smaller than 1.
\end{proof}



Lemma~\ref{f:full-3} holds for an arbitrary value of the input parameter $b$. Now,
in order to conclude the proof of Theorem~\ref{t:full-1}, we will show that, for a suitable choice of 
the constant parameter 
$b$, protocol \NAdaptiveU\ allows any station to transmit successfully whp within $br$ rounds, with $r = 4k\ln^2 k$.

\medskip
{\bf Proof of Theorem~\ref{t:full-1}:}
Consider any station $v$ and
let $r = 4\cdot k\ln^2 k$. 
By Lemma \ref{f:full-3}, there are at least $br/2$ rounds $t\in [1, br]$ in which 
$\hsigma[t]<1$.
In view of the inequality $\sigma[t] \leq \hsigma[t]$, it follows that
there are at least $br/2$ rounds $t\in [1, br]$ in which $\sigma[t]<1$.

By Lemma~\ref{Fact2}, the probability
that $v$ has a successful transmission at any of these rounds $t$ is at least
$q_v[t]/4 = p(t-t_v)/4 \ge p(br)/4$.
Hence, the probability that $v$ will not manage to transmit successfully in rounds $1,2,\ldots,br$, is at most
\[
\left(1-\frac{p(br)}{4}\right)^{br/2}
=
\left(1-\frac{\ln r}{4r}\right)^{br/2}
<
e^{-(b/8)\ln r}
<
k^{-b/8}
\ .
\]

This last value can be made smaller than 
$k^{-\eta-1}$, for any predefined parameter $\eta > 0$,
by choosing $b$ sufficiently large.
Applying the union bound over all contending stations 
we can derive that
the probability that any of them will not transmit successfully within 
the first $br$ rounds of its activity is less than $k^{-\eta}$.
This finally proves that our protocol \NAdaptiveU\ guarantees latency 
$br = O(k\ln^2 k)$ whp.
\qed

\subsection{Analysis of the protocol with acknowledgments.}\label{with_ack}

In this subsection we prove that if we allow each station to switch-off after getting an acknowledgment of its own successful transmission,
we can improve the 
performance guarantees of protocol \NAdaptiveU\ by a factor of $\Omega(\log\log k)$.


\begin{theorem}
\label{t:full-2}
All stations executing protocol \NAdaptiveU\, for a sufficiently large constant $b$, will transmit successfully
within $O\left(k\frac{\ln^2 k}{\ln\ln k}\right)$ time rounds whp. \gdm{The result holds even against an adaptive adversary.}
\end{theorem}

We start with the following technical fact.

\begin{fact}
\label{f:full-5}
Let $b_1 > 0$ be a constant.

If $r= \frac{2k\ln^2 k}{b_1\ln\ln k}$, then
$k\ln^2 r<  \frac{b_1r\ln\ln r}{2}$ for a sufficiently large $k$.

\end{fact}

\begin{proof}
%
For sufficiently large $k$, $\ln r < \ln k$, which implies that
$\frac{k\ln^2 r}{\ln\ln r} < \frac{k\ln^2 k}{\ln\ln k} = \frac{b_1}{2}  \frac{2k\ln^2 k}{b_1\ln\ln k} = \frac{b_1 r}{2}$.
Hence, $k\ln^2 r< \frac{b_1 r\ln\ln r}{2}$ for a sufficiently large~$k$.
\end{proof}


Likewise we did in Subsection \ref{no_ack}, we now fix a station $v$ with the intention of showing that such an arbitrary station will transmit successfully within $O\left(k\frac{\ln^2 k}{\ln\ln k}\right)$ rounds whp. All the following time rounds $t$ are 
refereed to a reference clock coinciding with $v$'s local clock.

\begin{lemma}
\label{f:full-7}
Let $b$ be an arbitrary positive integer constant.
There exists a constant $b_1 > 0$ such that for any round $t$ for which $1 < \sigma[t] \le b_1\ln\ln k$,
the probability that some station successfully transmits at round $t$ is at least $16k/(br)$,
where $r = \frac{2k\ln^2 k}{b_1\ln\ln k}$.
\end{lemma}

\begin{proof}
Fix any round $t$. 
The probability that some station transmits successfully at round $t$ is

\[\sum_{s\in A[t]} q_{s}[t] \prod_{s'\neq s\in A[t]} \left(1-q_{s'}[t]\right)
>
  \sigma[t] \left(\frac{1}{4}\right)^{\sigma[t]}
\ge
  b_1\ln\ln k \left(\frac{1}{4}\right)^{b_1\ln\ln k},
\]
where the first inequality follows from the observation that
 the transmission probabilities are less than $1/2$
	and the latter is implied by $f(x) = x\cdot4^{-x}$ being a decreasing
	function for $x \ge 1$.

Continuing, we have that
   \[
   b_1\ln\ln k \left(\frac{1}{4}\right)^{b_1\ln\ln k} 
   =\frac{b_1\ln\ln k} {2^{2b_1\frac{\log (\ln k)}{\log e}}}
   = \frac{b_1\ln \ln k}{(\ln k)^{2b_1/\log e}}.
   \]

Finally, by choosing $b_1$ sufficiently small and $k$ sufficiently large, the latter value can be made larger than
\[
\frac{4b_1\ln\ln k}{b\ln^2 k}=\frac{16k}{br}.
\]
\end{proof}

In the next lemma we prove the existence of sufficiently many rounds $t$ at which 
$\sigma[t]\leq b_1\ln\ln k$, for some constant $b_1>0$.
\gdm{Notice that, due to the adaptive adversary, these rounds are not known in advance, but the lemma guarantees that they will appear somewhere in the specified interval.}

\begin{lemma}
\label{f:full-6}
Let $b$ be an arbitrary positive integer constant. 
Suppose $r = \frac{2k\ln^2 k}{b_1\ln\ln k}$,
where $b_1$ is the constant determined in Lemma \ref{f:full-7}.
In at least $br/2$ rounds $t\in[1,br]$ the sum of probabilities $\sigma[t]$
is smaller than $b_1\ln\ln k$.
\end{lemma}

\begin{proof}
By Fact \ref{f:full-1} and Fact \ref{f:full-5}
we have that $k\cdot s({br}) < bk\ln^2 r < b\cdot \frac{b_1r\ln\ln r}{2}$.
Therefore,

\begin{eqnarray*}
	\sum_{t\in [1,br]} \hsigma[t] &=& \sum_{t\in [1,br]} \sum_{w\in \hat{A}[t]} q_w[t] \\
	&=&  \sum_{w\in \hat{A}[t]} \sum_{t\in [1,br]} p(t-t_w) \\ 
	&\le& k\cdot s(br)\\
	&<& \frac{br}{2}  \cdot b_1\ln\ln r.
\end{eqnarray*}

Thus,
in at least $br/2$ rounds $t\in[1,br]$ the value $\hsigma[t]$
is smaller than $b_1\ln\ln k$. Finally, the Lemma follows from the inequality $\sigma[t]\leq\hsigma[t]$.
\end{proof}

The next lemma shows that for a sufficiently large input parameter $b$, there are many rounds in which
the sum of transmission probabilities is at most 1 whp.

\begin{lemma}
\label{f:full-10}
Let $b,b_1,r$ be as in Lemma \ref{f:full-6}.
There exists $k_0$ not depending on $b,b_1,r$ such that for $k\geq k_0$,
in any execution of the protocol, the probability that $\sigma[t]\le 1$
 in at least $br/4$ rounds $t\in[1,br]$,  is at least
\[
         1 - \frac{1}{2k^{\eta+1}}.
\]

\end{lemma}

\begin{proof}
By Lemma \ref{f:full-6} there are at least $br/2$ rounds $t\in [1,br]$ in which $\sigma[t]\leq b_1\ln\ln k$. 
Let $R$ be the set of these rounds and
 $X$ be the event that there are at least $br/4$
rounds $t\in R$ such that $1<\sigma[t] \le b_1\ln\ln k$.
By Lemma \ref{f:full-7} 
in each of such rounds the probability 
of a successful transmission is at least $16k/(br)$.
To prove the lemma it will suffice to show that 
whp event $X$ will not occur.

	Let $T = t_1, t_2, \ldots, t_m$, with $t_1 < t_2 < \cdots < t_m$,  
	be the sequence of rounds in $R$
such that $1 < \sigma[t_i] \le b_1\ln\ln k$, for $i = 1,2,\ldots,m$.
	For any execution of protocol \NAdaptiveU, we can define the
	following binary sequence $\rho = \rho_1,\rho_2,\ldots, \rho_\ell$, 
	of length $\ell = br/4$ (all random choices are made independently).
\begin{enumerate}
\item If in round $t_i\in T$ there is no successful transmission, we set $\rho_i = 0$.
\item If in round $t_i\in T$ there is a successful transmission, then 
we set $\rho_i$ randomly so that $\Pr(\rho_i = 1) = 16k/(br)$. 
This is done as follows. Let 
$p_i = \sum_{v\in {A}[t_i]} q_v[t_i]\prod_{w\neq v} (1-q_{w}[t_i])$ be
the probability of having a successful transmission at round $t_i$,
we set
$$
\rho_i = \begin{cases}
1, & \text{with probability}\ (1/p_i) [16k/(br)]; \\
0, & \text{with probability}\ 1-(1/p_i) [16k/(br)].
\end{cases}
$$
\item If there are less than $br/4$ rounds in $T$, the lacking 
$\ell - m$ entries of $\rho$ are added by tossing a biased coin which assigns 1 
with probability $16k/(br)$ and 0 with probability $1 - 16k/(br)$.
\end{enumerate}

%
%
In this way, $\rho$ has length $\ell=br/4$ and each entry is independently 
set to 1 with probability \textit{exactly} $16k/(br)$.
Notice that at most $k$ 1's can be attributed by step (2), 
as there are at most $k$ stations and each of them can successfully transmit only
once, as it switches-off immediately after.
So, if there are more than $k$ 1's in the sequence $\rho$, then this exceeding amount of 1's must have been produced by the coin tosses in step (3). 
This implies that $m < \ell = (br)/4$ and therefore $X$ does not occur.
Hence, to finish the proof we need only to show that $\rho$ contains
more than $k$ 1's whp.  

The expected number of 1's in $\rho$ is $\mu =\ell \cdot 16k/(br) = 4k$.
Letting $\delta = 3/4$, the probability that the number of 1's in $\rho$
is at most $(1-\delta) \mu = k$ is, by the Chernoff bound, no more than
$e^{\delta^2\mu/2}=e^{-(3/4)^2 4k/2}$.
This value can be made smaller than 
$1/(2k^{\eta+1})$, for any predefined parameter $\eta > 0$,
by choosing $k$ sufficiently large, \textit{i.e.} for $k \ge k_0$,
where $k_0$ is a constant not depending on $b,b_1$ and $r$.
\end{proof}

Now we are ready to prove the main result of this section.

	\textbf{Proof of Theorem~\ref{t:full-2}:}
	Fix a parameter $\eta > 0$ and let $r = \frac{2k\ln^2 k}{b_1\ln\ln k}$,
	for some constant $b_1$ determined in Lemma \ref{f:full-7}.
	We will show that there exists an input constant
	$b$ for our protocol such that $v$ will transmit successfully within the
	first $br$ rounds of its activity.

	
	Let $S \subseteq [1,br]$ be a set of rounds, with $|S| \ge br/4$, 
	such that for every $t\in S$,  $\sigma[t]\le 1$. 
	
	By Lemma \ref{Fact2}, the probability of having a successful transmission 
	for $v$ in any $t\in S$ is at least 
	$q_v[t]/4 = p(t-t_v)/4 \ge p(br)$.
	Thus, the probability that $v$ will not send its message successfully 
	during all the rounds in $S$ is no more than
	\[
	(1-p(br)/4)^{br/4} = \left(1-\frac{\ln r}{4r}\right)^{br/4} < e^{-b\ln r/16}
	=  \left(\frac{2k\ln^2 k}{b_1 \ln\ln k}\right)^{-b/16} \le 
	k^{-b/16},
	\]
	where the last inequality holds for sufficiently large $k$.
	By chosing an input value $b$ sufficiently large, the latter value
	can be made less than $1/(2k^{\eta+1})$.

	By Lemma \ref{f:full-10}, 
	the probability that such a set $S$ does not exist
    is less than
	$1/(2k^{\eta+1})$, for $k$ sufficiently large. 

	Thus, the probability that $v$ does not transmit successfully during its first $br$ rounds is smaller than $1/k^{\eta+1}$.
	Finally, taking the union bound over all contending stations,
	it follows that all stations transmit successfully within 
	the first $br$ rounds of its activity with probability larger than $1-1/k^\eta$.
\qed	

\medskip
\gdm{
	We conclude the section with the following simple theorem on the energy cost of our protocol.
	
	\begin{theorem}
	\label{thm:energy-non-adaptive-unknown}
		The total number of broadcast attempts during the execution of Protocol \NAdaptiveU\  is $O(k\log^2 k)$ whp.
		The result holds even against an adaptive adversary.
	\end{theorem}
	\begin{proof}
\darek{Theorems~\ref{t:full-1} and \ref{t:full-2} guarantee that any station will transmit successfully within
$O\left(k\ln^2 k\right)$ rounds whp, regardless if acknowledgments are allowed or not.
By Fact~\ref{f:full-1}, the expected 
number of transmissions within such time interval (which is polynomial in $k$)
is $O(\log^2 k)$. Since the broadcast attempts of a single station  are independent over rounds,
by using the Chernoff bound for each station to upper bound the number of transmissions and then the union bound over all stations, 
we get $O(\log^2 k)$ transmissions for every station, whp.
This implies a total of $O(k\log^2 k)$ broadcast attempts whp.}
\end{proof}
}

\newcommand{\status}{\mbox{{\sc status}}}
\section{An adaptive algorithm for unknown contention}
\label{s:no-GC-no-k}

In Section~\ref{sec:lowerbound}
we have shown that, for non-adaptive algorithms, 
the lack of knowledge of contention's size, 
or even its linear approximation, 
makes the problem complexity provably higher.
However, in this section we prove that adaptive algorithms could overcome this obstacle.
We now describe a protocol {\tt AdaptiveNoK}, which resolves the contention
with linear latency without any knowledge of contention size.
\gia{Besides the data packet itself, each station can send 
a one-bit control message. For the sake of presentation we will
refer to these control messages as $<${\tt D mode}$>$ (encoded with bit 0) 
and $<${\tt any D-station left?}$>$ (encoded with bit 1).}

The algorithm works by alternating between two modes: 
{\em leader election mode} and {\em dissemination mode}.
The first mode aims at getting a synchronized subset of stations and electing a leader. The task of the leader will be to coordinate the computation in the dissemination mode, which aims at the actual contention resolution among the synchronized subset of stations defined in the previous execution of the leader election mode.
During the dissemination mode the leader has also the task to send periodically a message $<${\tt D mode}$>$
which informs \gia{newly awakened} stations about the mode currently executed.
A formal description of the modes and actions
can be found in the pseudocode of Protocol~\ref{alg:nswleader}.%
\footnote{%
Here we assume that a station does not automatically switch off after sending successfully,
but, taking advantage of being already selected, continues its activity coordinating the transmissions of the other stations. This is allowed due to the assumed adaptiveness of the
algorithm.
Moreover, the stations other than transmitting their own messages, can exchange other
information \dk{of limited (logarithmic) size; in our adaptive protocol, such control messages have only one bit.}
}
%
A \gia{newly awakened}
station first listens for $4$ rounds in order to determine the 
current mode of the system (cf. line \ref{while1}).
In the {\em leader election mode} ({\em L mode}) the stations execute a wake-up protocol, whose goal
is to get just one successful transmission. 
For this task, we use protocol {\tt DecreaseSlowly} introduced in~\cite{JS05}
(cf. the pseudocode of Protocol \ref{alg:decslowly}).
Once a successful transmission appears
in some round $t$, the station which transmitted in $t$ becomes the {\em leader},
it sets up a new variable 
\tc\ to $0$, and 
then the {\em dissemination mode} ({\em D mode}) starts.
At this point the leader and all other stations that were alive at $t$, are synchronized.
Let us denote by $C$ such a {\em synchronized} subset of stations
without 
\dk{a} leader. 
We can assume that a global clock 
(represented by variable \tc\ initiated by the leader)
starts for all the stations in $C$ at the round 
in which the leader was elected. This allows us to use any contention resolution
protocol for the synchronized \dk{(\textit{i.e.}, static)} model with unknown $k$ as a black-box -- we will call it {\tt SUniform}.

\begin{algorithm*}[t]
	\caption{{\tt AdaptiveNoK} (\textit{executed by a station $u$})}
	\label{alg:nswleader}
		\begin{algorithmic}[1]
		     \State {$\status \gets \emptyset$}
         \While {$\status \not = L$} \label{while1}
         \MLineComment{newly woken up stations can only enter the computation in $L$ mode}
             \State {listen to the channel for 4 rounds}
             \If {\gia{$u$ does not receive any message OR 
             it receives message $<${\tt is there anybody out there?}$>$}} \label{lmode}
                  \State {$\status \gets L$}
                   
             \EndIf
         \EndWhile
         \While {$u$ is active}
              \If {$\status = L$}
                  \State {execute {\tt DecreaseSlowly}} \label{decrease}
                  \MLineComment{the first successful station becomes the leader}
                  \State {$\tc \gets 0$} \MLineComment{now all awaken stations are synchronized and $\tc$ will denote the current round number started at the time the leader has been elected}
                  \State {$\status \gets D$} \MLineComment{once a leader has been elected the station switches to the dissemination mode}
              \EndIf
              \If {$\status = D$}
                   \If {$\tc$ is odd}
                       \If {$u$ is not the leader}
                              \State {execute {\tt SUniform$(u)$} (switch-off at the first successful transmission)} \label{suniformCall}
                       \EndIf
                   \ElsIf {$\tc = 2^x$ for some integer $x \ge 1$}
                        \State {transmit $<${\tt is there anybody out there?}$>$}\label{thewall}
                         \If {this transmission is successful $\And$ $u$ is the leader}
                             \State{ switch-off}\label{switch-lmode}
                              \MLineComment{if the leader gets an acknowledgement, it switches-off and the dissemination mode terminates}
                         \EndIf                        
                   \Else \label{else}
                        \If {$u$ is the leader} 
                                \State {transmit  $<${\tt D mode}$>$ \label{keep-dmode}}\MLineComment{as far as the leader is alive, the dissemination mode continues}
                        \EndIf
                   \EndIf
              \EndIf
         \EndWhile
    \end{algorithmic}
\end{algorithm*}

\begin{algorithm*}[t]
	\caption{{\tt DecreaseSlowly} (\textit{executed by a station $u$})}
	\label{alg:decslowly}
	\begin{algorithmic}[1]
    \State $q\gets$ some constant $>0$
    \State $i\gets $0
    \While {$u$ is active}
        \State transmit the message with probability $q\cdot \frac{1}{2q+i}$
        \If {transmission is successful}
            \State{become a leader}
        \EndIf
        \State $i\gets i+1$
    \EndWhile
    \end{algorithmic}
\end{algorithm*}

\paragraph{\darek{Protocol {\tt SUniform}}.}

In order to implement {\tt SUniform}, we can use one of the preexisting
\textit{Back-on/Back-off} protocols, which guarantee that all stations transmit successfully 
within $O(k)$ rounds after the synchronization round, whp.
\gia{As pointed out in \cite{Bend-05}, the idea of Back-on/Back-off 
(also known as Sawtooth Back-off) 
was discovered in other contexts many years ago \cite{sawtooth1,sawtooth2}.
This is a non-monotonic back-off strategy defining a sequence of 
contention windows within which any station chooses uniformly 
its transmitting round. 
It works in a doubly nested loop. 
The outer loop doubles the contention window at each step,
with the goal of ``guessing'' a window size proportional to the number 
of competing stations. For each such step, 
the innermost loop repeatedly fractions it until it reaches size 1,
so to progressively halving the number of active stations. The number of
such nested iterations is clearly $O(\log^2 T)$, where $T$ is the total
number of rounds until termination.
Ger\`eb-Graus and Tsantilas \cite{sawtooth1} considered the problem
of realizing arbitrary $h$-relations in a $n$-node network.
In an $h$-relation, each processor is both the source and the 
destination of at most $h$ messages. 
Setting $n = k$ and $h = k$, a solution to this problem 
can be used to solve our contention resolution among $k$ synchronized stations.
The protocol of \cite{sawtooth1} realizes an $h$-relation in a
$n$-node network in $\theta(h + \log n \log\log n)$ rounds
whp.
Hence, translated into our terminology and adapted to our needs, 
the result of \cite{sawtooth1} can be restated as follows
(see Theorem 4.2 in \cite{sawtooth1}). 

\begin{theorem}[\cite{sawtooth1}]\label{sawtooth}
Protocol {\tt SUniform} solves the contention among $k$ synchronized stations within $T$ rounds, where $T = O(k)$, whp, 
\textit{i.e.} with probability at least
$1 - 1/n^\alpha$ for any predefined constant $\alpha > 0$. 
It uses $O(\log^2 T)$ transmissions per station.
\end{theorem}
}

In order to use protocol {\tt SUniform} as a black-box,  
we need some additional synchronization mechanism 
inside the protocol itself and with respect to the main algorithm, 
which is described next.

First, during an execution of {\tt SUniform}, all stations from set $C$ 
switch off directly after a
successful transmission (are not active anymore), while all stations which 
are currently waking up
are waiting until the end of this particular execution of {\tt SUniform}.

Second,
in order to make possible that the stations can distinguish between both modes after waking up, the algorithm {\tt SUniform} is executed in odd rounds only, while even rounds are devoted to performing the following kind of coordination.

\gia{
In rounds expressible as $2^x$ for integers $x > 1$,  
all alive stations from set $C$ and the leader 
transmit together the message $<${\tt is there anybody out there?}$>$. 
In the remaining even rounds, the leader broadcasts the message $<${\tt D mode}$>$. 
The latter transmission, which is surely successful as it is performed by
the leader only, informs the newcomers to stay silent until the dissemination
mode hasn't finished. In order to understand when the $D$ mode has terminated,
the outcome of the former transmission 
(message $<${\tt is there anybody out there?}$>$), is used. 
There are two cases.

1) If the leader receives an acknowledgment on such a transmission, 
then it knows that it was the only transmitter 
(all stations in $C$ successfully transmitted and switched off). In this case,
the leader itself switches-off and the $<${\tt D mode}$>$ message is suspended  (cf. line \ref{switch-lmode}).

2) If the leader does not receive an acknowledgment, it knows that there are still some stations in $C$ which have not succeeded in transmitting a message. Consequently, the leader keeps transmitting the 
$<${\tt D mode}$>$ message (cf. line \ref{keep-dmode}).

Notice that at most 4 consecutive rounds 
would be needed to the leader, as far as it is still alive, to
send the message $<${\tt D mode}$>$ (one even round might be skipped
if it is a power of 2 larger than 2).

\smallskip
Now we can see how the control messages guarantee an alternation between the 
two modes until no new station arrives in the system, which causes the algorithm
to stop.
}

\paragraph{\darek{Controlling modes.}}
\gia{

When a first group of stations is activated (so that there is no station already active in the system), no message can be received during the execution of the first while loop, and so an $L$ mode is started by these stations:
they exit the loop and execute protocol {\tt DecreaseSlowly} to elect a leader.
Any station that possibly wakes up during this $L$ mode, either
perceives the successful transmission of the leader election 
during the 4 rounds of waiting or nothing.
The condition in line \ref{lmode} guarantees
that in the first case the station keeps waiting, while the latter
case causes the station to exit the loop and join the leader election.

Once a leader has been elected, all (and only) the stations that participated 
to the election, start a $D$ mode. Any station that possibly 
wakes up during this $D$ mode, keeps waiting in the first while loop
until it gets a $<${\tt D mode}$>$ message. 
Once the $D$ mode terminates, the leader receives an acknowledgement on
its transmission $<${\tt is there anybody out there?}$>$.
In such a case, all newcomers that were waiting in
the while loop, either receive this 
acknowledgement or (if they were activated just after) they don't 
receive a $<${\tt D mode}$>$ message. 
In both cases, they learn that the dissemination mode has terminated
and the condition in line \ref{lmode} guarantees that
they exit the loop and start a new leader election.

This process is iterated until no new station is injected in the system, 
which happens when a $D$ mode finishes and no station is waiting
in the while loop.
}

\subsection{Analysis}

The correctness and the time complexity of this algorithm will be proved in Theorem~\ref{t:withLeader}.
The first step is to show that protocol {\tt DecreaseSlowly} wakes up the system in $O(k)$ rounds whp. 
In \cite{JS05}, an algorithm {\em Decrease Slowly} has been presented which completes the
wake-up in $O(k\log k)$ rounds whp. In the following theorem, 
we improve the analysis and show that actually this algorithm
can complete the wake-up in $O(k)$ rounds whp.

\begin{theorem}\label{t:wakeup}
Algorithm {\tt DecreaseSlowly} finishes wake-up in $O(k)$ rounds whp.
This results holds even for an adaptive adversary.
\end{theorem} 

\begin{proof}
Let us consider the first $32qk$ rounds following the round at which the first station wakes up
and starts the computation. We will prove that by the end of this interval the wake up has been
accomplished whp.

Following the algorithm, each awake station, starting from the round at which it wakes up, transmits 
with probability $q\cdot \frac{1}{2q}, q\cdot \frac{1}{2q+1}, q\cdot \frac{1}{2q+2}, \ldots$
If we denote by $p_i$ the transmission probability of an arbitrary awake station $u$ at the $i$th round of 
its computation, we have that the sum of transmission probabilities of $u$ over $32qk$ rounds is
\begin{equation}\label{ssum}
  s(32qk) = \sum_{i=0}^{32qk} p_i \leq q \left( \sum_{i=0}^{32qk} \frac{1}{i} \right) \leq q (1+\ln(32qk)) \ ,
\end{equation}
where in the last step we have used the known bounds for the $\delta$th partial sum $H_{\delta}$ of the harmonic series:
\begin{equation}\label{harmonic}
    \ln(1+h) \leq H_{\delta} = \sum_{i=1}^{h} \frac{1}{i} \leq 1 + \ln h \ .
\end{equation}

For any fixed round $t$, let us consider now the sum of transmission
probabilities of all awake stations at time $t$, denoted as in the previous section by $\sigma[t]$.
Recalling that $A[t]$ is the set of active stations at round $t$, the probability that a station successfully transmits at round $t$ is
\begin{eqnarray}
    \sum_{v\in A[t]} q_{v}[t] \cdot \prod_{w\neq v} (1-q_{w}[t]) 
            &=&   \sum_{v\in A[t]} q_{v}[t] \cdot \prod_{w\neq v} (1-q_{w}[t])^{\frac{1}{q_w[t]}q_{w}[t]} \nonumber \\
            &\geq&  \sum_{v\in A[t]} q_{v}[t] \cdot \prod_{w\neq v} (1/4)^{q_{w}[t]} \nonumber \\
            &>& \sigma[t] \cdot (1/4)^{\sigma[t]} .\label{sigmale}
\end{eqnarray}

Since at most $k$ stations can be awake in each round, \darek{for any adversarial wake-up strategy,} the average sum $\sigma[t]$ for $t$ ranging
over our interval of $32qk$ rounds will be
\[
   \frac{1}{32qk} \sum_{t=0}^{32qk} \sigma[t] \leq \frac {q (1+\ln(32qk))\cdot k}{32qk} \leq \frac{\ln(32qk)}{16} \ .
\]
Of course, in at least half of the interval, $\sigma[t]$ must be not larger than 
twice the average; therefore there is a set $T$ of at least $16qk$ rounds such that 
\begin{equation}\label{eq:ub}
\sigma[t] \leq \frac{\ln(32qk)}{8} \ ,
\end{equation}
for every $t\in T$.
Let us consider only rounds in $T$. We say that a round $t$ is {\em heavy} when 
$\sigma[t] > 1/2$ and {\em light} otherwise.
We distinguish two  
\darek{complementary cases, thus in each execution caused by an adaptive adversary one of them must occur. We will show that in each of these cases, no matter how the heavy and light rounds are distributed, the wake-up  occurs whp. Then, summing up the conditional probabilities over all possible distributions of heavy and light rounds, and by taking the union bound of the two cases, the theorem holds. Hence, it remains to analyze these two cases.}
%

\begin{enumerate}
\item
There are at least $8qk$ heavy rounds. Recalling also (\ref{eq:ub}),
in at least $8qk$ rounds $t$, it holds that 
\begin{equation}\label{eq:hypo}
  \frac{1}{2} < \sigma[t] \leq \frac{\ln(32qk)}{8} \ .
\end{equation}
Plugging (\ref{eq:hypo}) into (\ref{sigmale}), we get that
the success probability at round $t$ is at least
\[
    \frac{1}{2}\cdot \left(\frac{1}{4}\right)^{\frac{\ln(32qk)}{8}} \ge \left( \frac{1}{16} \right)^{\frac{\ln(32qk)}{8}} \ge \frac{1}{\sqrt{32qk}} \ .
\]
Therefore, the probability that the wake-up does not appear in $8qk$ heavy rounds is at most
$(1-1/\sqrt{32qk})^{8qk} = O(1/k^a)$ for an arbitrary constant $a$ depending on $q$.

\item
There are less than $8qk$ heavy rounds. So, there are at least $\delta = 8qk$ light rounds.
Let $t_1,t_2,\ldots,t_{\delta}$ be the time-ordered sequence of light rounds.
It is possible to show by induction on $i$  
(cf. Claim 1 in the proof of Theorem 10.3)  
that $\sigma[t_i] \geq q / (2q+i)$, for $1\le i \le \delta$. Consequently, 
\[
                    q / (2q+i) \leq  \sigma[t_i] \leq 1/2, \mbox{ for } 1\le i \le \delta. 
\]
Plugging these bounds into (\ref{sigmale}) we get that the
probability of successfully waking up in the $i$th light round is at least 
$q/2(2q+i)$. Hence, the probability that the wake-up is not successful is at most
\begin{eqnarray*}
  \prod_{i=1}^{\delta} (1 - \sigma[t_i]) 
            & \leq & \prod_{i=1}^{\delta} \left( 1 - \frac{q}{2(2q + i)}\right) \\
            & \leq & \left( \frac{1}{e} \right) ^{\sum_{i=1}^{\delta}\frac{q}{2(2q + i)}} 
              \leq   \left( \frac{1}{e} \right) ^{\frac{q}{2} (\sum_{i=1}^{\delta}\frac{1}{i} - \sum_{i=1}^{2q}\frac{1}{i}) }\\          
            & =    & \left( \frac{1}{e} \right) ^{\frac{q}{2} (H_{\delta} - H_{2q}) } 
              \leq   \left( \frac{1}{e} \right) ^{\frac{q}{2} (\ln(1+\delta) - \ln(4q))} \;\;\; \mbox{ (by (\ref{harmonic}))}\\
            & = & \left( \frac{4q}{1+8qk} \right) ^{q/2} 
             \leq \left( \frac{1}{2k} \right)^{q/2}.
\end{eqnarray*}
\end{enumerate}
\end{proof}

The following lemma shows that once a station enters the computation in $L$ mode, 
i.e. after it has left the first while loop in line \ref{while1} of Protocol {\tt AdaptiveNoK}, 
it reaches a successfull transmission within $O(k)$ rounds whp. The final theorem 
will further consider the time spent inside the first while loop.

\begin{lemma}\label{l:while}
A station $u$ sends successfully whp in $O(k)$ rounds after exiting the while loop in line \ref{while1} of
Protocol {\tt AdaptiveNoK}.
\end{lemma}

\begin{proof}
A station $u$ that exits the while loop in line \ref{while1} is in $L$ mode.
In $L$ mode, the station executes Protocol {\tt DecreaseSlowly} which, by Theorem~\ref{t:wakeup},
allows \textit{one of the participating stations} to transmit successfully to the channel within $O(k)$ rounds whp.
The successful station becomes the leader and the status changes to $D$ mode. 
Obviously, if $u$ is such a station, then we have proved the lemma.
Therefore, suppose that $u$ is not the leader. 
Starting from the time at which the leader has been elected, all
stations that were active at that time are synchronized, i.e., they can start a clock at the time 
of the leader election (round $\tc = 0$). Let $v \not = u$ be the leader and $A$ be the set of all 
other active stations synchronized at round 0 (according to the \tc).
In such a situation, $u$ and all other stations in $A$ execute 
(in rounds in which \tc\ is odd) 
\gia{protocol SUniform} guaranteeing by Theorem \ref{sawtooth} 
contention resolution
in $O(k)$ time whp for the static model, \textit{i.e.}, when all participating stations start at the same time, which is
our situation. This concludes the proof.
\end{proof}

Finally, we are ready for the main theorem of this section.

\begin{theorem}\label{t:withLeader}
Algorithm {\tt AdaptiveNoK} solves the contention resolution problem 
with latency $O(k)$, 
whp.
The result holds even against an adaptive adversary.
\end{theorem}

\begin{proof}
Our aim is to show that any fixed station $u$ executing {\tt AdaptiveNoK}
transmits successfully its message within $O(k)$ rounds since its wake up time, whp. 
In view of Lemma~\ref{l:while}, it is sufficient to
show that the station exits the while loop in line \ref{while1} 
within $O(k)$ rounds.
\gia{
Two cases can occur. 

If $D$ mode is not running, the station exits the loop within 4 rounds.
Indeed, in such a case, either the station hasn't received any message 
or it has got the message $<${\tt is there anybody out there?}$>$. 
This corresponds to the condition in line \ref{lmode} which causes the
station to leave the while loop.

If $D$ mode is running, the station stays in the while loop as long as it 
keeps receiving message $<${\tt D mode}$>$.
}
Let us call \textit{white rounds} the even rounds expressible as $2^x$ 
for integers $x > 1$, and \textit{black rounds} all the remaining even rounds.
The message $<${\tt D mode}$>$ is sent by the leader when it executes 
line \ref{keep-dmode}, which happens in every black round, unless it 
got an acknowledgement in the last execution of line \ref{switch-lmode}
and switched off.
Being two consecutive black rounds at most 4 rounds apart, it will suffice
to wait at most 4 rounds for a newly woken up station to establish whether 
the leader is still sending its message $<${\tt D mode}$>$.
Since the delay between the acknowledgement and the first \gia{subsequent} 
black round
is at most 2 rounds, it will be sufficient to show that the leader gets
the acknowledgement after $O(k)$ rounds.
Notice that such an acknowledgement is perceived if and only if the leader is the
only sender, that is, if and only if all stations that started 
the execution of {\tt SUniform} have switched off.
Therefore, the leader will get the acknowledgement in
the first white round following the switching off of all stations running 
{\tt SUniform}, which happens after at most twice the running time of
{\tt SUniform}.

The proof now follows by observing that,
\gia{by Theorem \ref{sawtooth}}, 
after $O(k)$ rounds whp
all the stations executing {\tt SUniform} have switched off.     

\end{proof}

\gdm{
We now conclude the section by deriving the energy cost of our adaptive algorithm. 

\begin{theorem}
\label{thm:energy-adaptive}
The \gia{expected} total number of broadcast attempts during the execution of protocol {\tt AdaptiveNoK} is $O(k\log^2 k)$.
\darek{The result holds even against an adaptive adversary.}
\end{theorem}
}

\begin{proof}
\gia{
Starting in $L$ mode, the system alternates between $L$ mode and 
$D$ mode until it happens that no newcomer arrives during the last $D$ mode.
Therefore, the entire round sequence of the algorithm, from start
to finish, can be partitioned into $\tau > 0$ disjoint intervals
$L_1,D_1, L_2, D_2, \ldots, L_{\tau}, D_{\tau}$,
where $L_j$ (respectively $D_j$) is the time interval within which the system
is involved in the $j$th $L$ mode (respectively $D$ mode) and
$D_{\tau}$ is the time interval of the last $D$ mode, \textit{i.e.} 
the one such that no new station arrives during its execution.

For $1 \le j \le \tau$, let $S_j$ be the set of stations 
active during the interval $L_j, D_j$. This is the set of stations that
participate to the $j$th leader election and to the subsequent dissemination
mode. Recalling that each of these stations permanently switches off by the
end of this $D$ mode, the sequence $(S_1,S_2, \ldots, S_{\tau})$ form a 
partition of the total set of $k$ stations. 
To prove the theorem we will  
show that the expected total number of transmissions needed during
the time interval $L_j,D_j$ is $O(\kappa \log^2 \kappa)$, where
$\kappa = |S_j|$. By the linearity of expectation, the theorem follows.
Let us consider the execution during an arbitrary $L_j,D_j$ time interval. 
We count the transmissions needed in $L_j$ and $D_j$ separately.

\paragraph{Transmissions in $L_j$.}
For an arbitrary set of $\kappa$ stations executing this $L$ mode, 
let $X(\kappa)$ and $Y(\kappa)$ be the random variables denoting 
respectively the total number of transmissions spent by all stations
during the time interval $L_j$ and the size of $L_j$, 
\textit{i.e.}, the number of rounds
until the first successful transmission in protocol {\tt DecreaseSlowly}.
The expected total number of transmissions is
\begin{eqnarray}\label{expected}
\E\big(X(\kappa)\big) 
&=& 
\sum_{i=1}^{\infty} 
\Pr\Big( \Omega(\kappa^{i-1}) \le Y(\kappa) \le o(\kappa^i)\Big)\;
\E\Big(X(\kappa) \big|\; \Omega(\kappa^{i-1}) \le Y(\kappa) \le o(\kappa^{i})\Big)\nonumber \\
&\le&
\sum_{i=1}^{\infty} 
\Pr\Big(Y(\kappa) \ge \Omega(\kappa^{i-1})  \Big) \;
\E\Big(X(\kappa) \big|\; Y(\kappa) \le O(\kappa^{i})\Big)
\ .
\end{eqnarray}

We can now prove by induction that for any $i\ge 1$,
\begin{equation}\label{firstbound}
\Pr\Big(Y(\kappa) \ge \Omega(\kappa^{i-1})  \Big) 
\le
\Pr\Big(Y(\kappa^{i-1}) \ge \Omega(\kappa^{i-1})  \Big) 
= O\left(\frac{1}{\kappa^{i-1}}\right)
\ .
\end{equation}
For $i = 1$ the assertion is obvious. 
For $i \ge 2$, if protocol {\tt DecreaseSlowly} needs 
$\Omega(\kappa^{i-1})$ rounds on an input instance of 
$\kappa$ stations even more so on an input instance of 
$\kappa^{i-1} \ge \kappa$ stations. 
That is, 
$\Pr\Big(Y(\kappa) \ge \Omega(\kappa^{i-1}) \Big) 
\le \Pr\Big(Y(\kappa) \ge \Omega(1) \Big)$, for $i \ge 2$.
The equality of (\ref{firstbound}) derives from Theorem \ref{t:wakeup}, which on
an instance of $\kappa^{i-1}$ stations guarantees that the protocol terminates
within $O(\kappa^{i-1})$ rounds with high probability.

If we are given the information that the protocol terminates 
within $O(\kappa^i)$ rounds, then the expected 
total number of transmissions can be obtained as follows,
\begin{equation}\label{secondbound}
\E\Big(X(\kappa) \big|\; Y(\kappa) \le O(\kappa^{i})\Big) = 
\kappa \cdot \sum_{i=0}^{O(\kappa^i)} q/(2q+i) = O(\kappa \log (\kappa^i))
\ ,
\end{equation}
where the equality follows by (\ref{ssum}) and (\ref{harmonic}). 
Hence, plugging (\ref{firstbound}) and (\ref{secondbound}) into (\ref{expected}),
we get
\[
\E\big(X(\kappa)\big) 
\le
\sum_{i=1}^{\infty} 
O\left(\frac{1}{\kappa^{i-1}}\right)
O(i\cdot\kappa\log \kappa) = O(\kappa\log \kappa)
\ .
\]

\paragraph{Transmissions in $D_j$.}
Here we need to count the transmissions required by protocol {\tt SUniform}
plus those spent by the leader and the other 
synchronized stations during the dissemination mode.
If we let $X(\kappa)$ and $Y(\kappa)$ relate to $D_j$ 
(rather than $L_j$), we can observe that the expected
number of transmissions per station is again described by (\ref{expected}).
The running time of this $D$ mode is dominated by the number of rounds
of {\tt SUniform}. Hence, 
the upper bound on (\ref{firstbound}) is also preserved, as
Theorem \ref{sawtooth} guarantees that 
{\tt SUniform} on an instance of $\kappa^{i-1}$ stations 
terminates within $O(\kappa^{i-1})$ rounds after the synchronization 
round, whp.

Theorem \ref{sawtooth} also states that if the protocol terminates
within $O(\kappa^i)$ rounds, then the number of transmissions 
spent by any station is $O(\log^2 (\kappa^i))$.
If the protocol terminates within $O(\kappa^i)$ rounds,
then the leader transmits $O(k^i)$ times and the other 
synchronized stations involved in the dissemination
mode, transmit only in rounds expressible as powers of two,
which are $O(\log (\kappa^i))$.
Hence,  (\ref{secondbound}) in this case writes as
\[
\E\Big(X(\kappa) \big|\; Y(\kappa) \le O(\kappa^{i})\Big) = 
 O(\kappa^i + \kappa \log^2 (\kappa^i) + \kappa\log (\kappa^i) )
 \ .
\]
Now, if we choose $\alpha \ge 2$ in Theorem \ref{sawtooth},
we can make the probability of failure for {\tt SUniform},
on an instance of $\kappa^{i-1}$ stations, at most $1/k^{2(i-1)}$.
Consequently,
\[
\E\big(X(\kappa)\big) 
\le
\sum_{i=1}^{\infty} 
\frac{1}{\kappa^{2(i-1)}}
O\left(\kappa^i + \kappa \log^2 (\kappa^i) + \kappa\log (\kappa^i) \right)
= \kappa \log^2 \kappa
\ .
\]
\medskip
This proves that the expected total number of transmissions spent 
during the interval of rounds $L_j,D_j$ is $O(|S_j| \log^2 |S_j|)$, 
for $j=1,2,\ldots, \tau$. Recalling  
that $(S_1,S_2, \ldots, S_{\tau})$ is a partition of the set of $k$ stations, by the linearity of expectation, the theorem 
follows.
}
\end{proof}

\section{Discussion and open problems}

\dk{Although our algorithms reach optimal or almost optimal latency and transmission energy cost, one could ask about their efficiency in terms of the number of listening slots. Non-adaptive algorithms have the clear advantage that they do not need to listen to the channel. As for our adaptive algorithm, its worst-case performance in terms of listening slots is incurred by newly awaken stations in the very beginning of the algorithm, in particular, a single such station could spend even $\Theta(k)$ slots listening and waiting to change its mode to $L$. Moreover, the number of such awaiting stations could be $\Theta(k)$, therefore even an amortized number of listening slots per process could be as high as $\Theta(k)$.
Reducing this cost without harming latency or transmission energy is a challenging open problem, as it also limits information flow. We conjecture, however, that such a reduction to polylogarihmic formula per station should be doable.}

\dk{Another} interesting open direction is to study trade-offs between energy consumption and other measures.
In particular, is logarithmic energy necessary for anonymous shared channel against an adaptive adversary
in order to achieve 
\darek{the minimum possible (asymptotically) latency $O(k)$?} 
If so, could we lower the energy requirement by allowing slightly
\darek{larger latency, and if so, how much larger?}

\dk{There are also other model features that may influence performance of contention resolution. For instance, the availability of a global clock may improve synchronization, but could it be asymptotically better than the more natural setting without global clock?
\gs{If, \textit{e.g.},\ the stations have access to a global clock and all stations get acknowledgments of all transmissions, they can easily solve the contention resolution problem
with latency $O(k)$. They can do it by using the wakeup UFR algorithm from paper \cite{JS05}.
Wakeup is performed in odd rounds and in even rounds all stations transmit with the probability
from the last successful wakeup round. Every station switches off after transmitting its message successfully. This approach should assure maintaining optimal transmission probabilities of stations for a constant fraction of active time.
Will global clock be also useful, if the transmitting station only gets acknowledgements?
}
On the other hand, failures or external interference may substantially worsen the performance, depending on their severity and intensity -- as it does in other related settings (c.f., job scheduling on wireless channel under jamming and failures~\cite{AntaGKZ17,KlonowskiKMW19}).}

\bibliographystyle{plain}
\bibliography{bibliography}


\end{document}